\documentclass[letterpaper,11pt]{article}

\usepackage[margin=1in]{geometry}                
\usepackage{amssymb}
\usepackage{setspace}
\usepackage{authblk}
\usepackage[title]{appendix}\usepackage{amsmath}%
\usepackage{amsfonts}%
\usepackage{amssymb}%
\usepackage{amsthm}
\usepackage{breqn}
\usepackage{graphicx}
\usepackage{colonequals}
\usepackage{psfrag}
\usepackage{subfigure}
\usepackage{mathrsfs}
\usepackage[noadjust]{cite}
\usepackage{hyperref}

\usepackage[usenames,dvipsnames]{xcolor}
\usepackage{dsfont}
\usepackage{amsbsy}

\allowdisplaybreaks

\newtheorem{thm}{Theorem}
\newtheorem{lem}{Lemma}

\newtheorem{prop}{Proposition}

\newtheorem{cor}{Corollary}

\theoremstyle{definition}

\newtheorem{defn}{Definition}
\newtheorem{scheme}{Scheme}
\newtheorem{remark}{Remark}
\newtheorem{example}{Example}

\newcommand{\etal}{\textit{et al.}~}
\newcommand{\Fq}{\mathbb{F}_{q}}
\newcommand{\PR}{\mbox{Pr}}

\newcommand{\calX}{\mathcal{X}}
\newcommand{\calA}{\mathcal{A}}
\newcommand{\calB}{\mathcal{B}}
\newcommand{\calK}{\mathcal{K}}
\newcommand{\calM}{\mathcal{M}}
\newcommand{\calY}{\mathcal{Y}}

\newcommand{\calJ}{\mathcal{J}}
\newcommand{\calC}{\mathcal{C}}
\newcommand{\calS}{\mathcal{S}}

\newcommand{\calL}{\mathcal{L}}

\newcommand{\Xn}{\mathcal{X}^n}
\newcommand{\cwd}{\left[2^{nR} \right]}
\newcommand{\floor}[1]{\lfloor #1 \rfloor}
\newcommand{\ceil}[1]{\lceil #1 \rceil}

\newcommand{\bc}{\mathbf{c}}

\newcommand{\bH}{\mathbf{H}}

\newcommand{\bD}{\mathbf{D}}
\newcommand{\bx}{\mathbf{x}}
\newcommand{\rank}{\mbox{rank}}
\newcommand{\bz}{\mathbf{z}}

\renewcommand{\tilde}{\widetilde}
\newcommand{\Enc}{\mathsf{Enc}}
\newcommand{\Dec}{\mathsf{Dec}}

\newcommand{\Reals}{\mathbb{R}}

\newcommand{\defined}{\triangleq}

\newcommand{\ExpVal}[2]{\mathbb{E}\left[ #2 \right]}

\newcommand{\by}{\mathbf{y}}

\newcommand{\blambda}{\pmb{\lambda}}

\newcommand{\bu}{\mathbf{u}}

\newcommand{\ba}{\mathbf{a}}
\newcommand{\bb}{\mathbf{b}}

\newcommand{\EE}[1]{\ExpVal{}{#1}}

\newcommand{\bsigma}{\boldsymbol\sigma}

\newcommand{\brho}{\boldsymbol{\rho}}
\newcommand{\suchthat}{\,\mid\,}
\newcommand{\indicator}{\mathds{1}}
\newcommand{\mmse}{\mathsf{mmse}}
\newcommand{\error}{\mathsf{e}}

\newcommand\independent{\protect\mathpalette{\protect\independenT}{\perp}}
\def\independenT#1#2{\mathrel{\rlap{$#1#2$}\mkern2mu{#1#2}}}

\usepackage{prettyref,enumerate}

\title{Hiding Symbols and Functions:  \\ New Metrics and Constructions for Information-Theoretic
Security}
\author{Flavio P. Calmon,  Muriel M\'edard, Mayank Varia, Ken R. 
  Duffy,\\ Mark M.
  Christiansen, Linda M. Zeger \thanks{Some of the results in this paper
  were presented at the 50th and 52nd Allerton Conference on
Communications, Control and Computing \cite{inproc:allertonCrypt,calmon2014allerton}.
\\
F.~P.~Calmon and M.~M\'edard are with the Research Laboratory
of Electronics at the Massachusetts Institute of Technology, Cambridge, MA (email:
flavio@mit.edu; medard@mit.edu).
\\
M.~Varia is with the MIT Lincoln Laboratory, Lexington, MA (e-mail:
mayank.varia@ll.mit.edu).
\\
 M.~M.~Christiansen and K.~R.~Duffy are with the Hamilton Institute, Maynooth
 University, Maynooth, Co Kildare, Ireland (e-mail:
mark.christiansen@nuim.ie; ken.duffy@nuim.ie).
\\
L. M. Zeger is currently with Auroral LLC, and was with the MIT Lincoln Laboratory, Lexington, MA, zeger@auroral.biz.
\\
F.~P.~Calmon and M.~Varia were sponsored by the Intelligence Advanced Research Projects
Activity under Air Force Contract FA8721-05-C-0002.  Opinions, interpretations,
conclusions and recommendations are those of the author and are not necessarily
endorsed by the United States Government.
}
}
\date{}                                           

\begin{document}

\maketitle
\vspace{-.2in}

\begin{abstract}
   
  We present  information-theoretic definitions and results for analyzing
  symmetric-key encryption schemes beyond the perfect secrecy regime, i.e. when
  perfect secrecy is not attained. We adopt two lines of analysis, one based on
  lossless source  coding,  and another akin to rate-distortion theory. We start
  by presenting a new information-theoretic metric for security, called
  $\epsilon$-symbol secrecy, and derive associated fundamental bounds. This
  metric provides a parameterization of secrecy that spans other
  information-theoretic metrics for security, such as weak secrecy and perfect
  secrecy.  We then introduce  list-source codes (LSCs), which are a general
  framework for mapping a key length (entropy) to a list size that an
  eavesdropper has to resolve in order to recover a secret message.  We provide
  explicit constructions of LSCs, and show that LSCs that achieve high symbol
  secrecy also achieve a favorable  tradeoff between key length and uncertainty
  list size. We also demonstrate that, when the source is uniformly distributed,
  the highest level of symbol secrecy for a fixed key length can be achieved
  through a construction based on  minimum-distance separable (MDS) codes.
  Using an analysis related to rate-distortion theory, we then show how symbol
  secrecy can be used  to determine the probability that an eavesdropper
  correctly reconstructs functions of the original plaintext. More specifically,
  we present lower bounds for the minimum-mean-squared-error of estimating a
  target function of the plaintext given that a certain set of functions of the
  plaintext is known to be hard (or easy) to infer, either by design of the
  security system or by restrictions imposed on the adversary. We illustrate how
  these bounds can be applied to characterize  security properties of
  symmetric-key encryption schemes, and, in particular, extend security claims
  based on symbol secrecy to a functional setting. Finally, we discuss the
  application of our methods in key distribution, storage and privacy.

\end{abstract}

\newpage
\tableofcontents

\onehalfspacing

\section{Introduction}

The security properties of a communication scheme can, in general, be evaluated
from two fundamental perspectives:  information theoretic  and computational.
For a noiseless setting, unconditional (i.e. perfect) information-theoretic
secrecy can only be attained when the communicating parties share a random key
with entropy at least as large as the message itself
\cite{shannon_communication_1949}. Consequently, usual information-theoretic
approaches focus on physically degraded models \cite{liang_information_2009},
where the goal is to maximize the secure communication \textit{rate} given that
the adversary has a noisier observation of the message than the legitimate
receiver.  On the other hand, computationally secure cryptosystems have thrived
both from a theoretical and a practical perspective.  Such systems are based on
yet unproven hardness assumptions, but nevertheless have led to cryptographic
schemes that are widely adopted (for an overview, see
\cite{katz_introduction_2007}).  Currently, computationally secure encryption
schemes  are used millions of times per day, in applications that range from
online banking transactions to digital rights management. 



Computationally secure cryptographic constructions do not necessarily provide an
information-theoretic guarantee of security. For example, one-way permutations
and public-key encryption cannot be deemed secure against an adversary with
unlimited computational resources. This is not to say that such primitives are
not secure in practice -- real-world adversaries are indeed computationally
bounded. There are, however, cryptographic schemes that are believed to be
computationally secure and simultaneously provide  \textit{some} security
guarantee against computationally unbounded adversaries, albeit such guarantee
is not absolute secrecy. This was noted by Shannon
\cite{shannon_communication_1949} and later by Hellman
\cite{hellman_extension_1977} in a companion paper to his and Diffie's work
``New directions in Cryptography'' \cite{diffie_new_1976}. 



Our goal in this work is to characterize the fundamental information-theoretic
security properties of cryptographic schemes when perfect secrecy is not
attained. We follow the  footsteps of Shannon and Hellman and study
symmetric-key encryption with small keys, i.e. when the length of the key is
smaller than the length of the message. In this case, the best a computationally
unrestricted adversary can do is to decrypt the ciphertext with all possible
keys, resulting in a list of possible plaintext messages. The adversary's
uncertainty regarding the original message is then represented by a probability
distribution over this list. This distribution, in turn, depends on both the distribution of
the key and the distribution of the plaintext messages. 

We evaluate the information-theoretic  security in this setting  through two
complementary  lines of analysis: (i) one  based on lossless source coding,
where the security properties of the uncertainty list are measured using mutual
information-based metrics and secure communication schemes are provided based on
 linear
code constructions, and (ii) another akin to rate-distortion theory, where the
mutual information-based metrics are translated into restrictions on the
inference capabilities of the adversary through converse results. We
describe each approach below.

\subsection{Lossless Source Coding Approach}

If perfect secrecy is not achieved, then meaningful metrics are required to
quantify the level of information-theoretic security provided by a cryptographic
scheme. We  define a new metric for characterizing security,
$\epsilon$-\textit{symbol secrecy}, which quantifies the uncertainty of specific
source symbols given an encrypted source sequence.  This metric subsumes
traditional rate-based information-theoretic measures of secrecy which are
generally asymptotic \cite{liang_information_2009}. However, our
definition is not asymptotic and, indeed, we provide a construction that
achieves fundamental symbol secrecy bounds, based on maximum distance separable
(MDS) codes, for finite-length sequences. We note that there has been a long
exploration of the connection between coding and cryptography
\cite{blahut_communications_1994}, and our work is inscribed in this school of
thought. 

We also introduce a general source coding framework for analyzing the
fundamental information-theoretic properties of symmetric-key encryption, called
\textit{list-source codes} (LSCs). LSCs compress a source sequence
\textit{below} its entropy rate and, consequently, a message encoded by an LSC
is decoded to a list  of possible source sequences instead of a unique source
sequence. We demonstrate how any symmetric-key encryption scheme can be cast as
an LSC, and prove that the best an adversary can do is to reduce the set of
possible messages to an exponentially sized list with certain properties, where
the size of the list depends on the length of the key and the distribution of
the source. Since the list has a size exponential in the key length, it cannot
be resolved in polynomial time in the key length, offering a certain level of
computational security. We characterize the achievable  $\epsilon$-symbol
secrecy of LSC-based encryption schemes, and provide explicit constructions
using  algebraic coding. 



\subsection{Rate-Distortion Approach}
While much of information-theoretic security has considered the hiding of the
plaintext, cryptographic metrics of security seek to hide also functions thereof
\cite{goldwasser_probabilistic_1984}. More specifically,  cryptographic metrics
characterize how well an adversary can (or cannot) infer functions of a hidden
variable, and are  stated in terms of lower bounds for average estimation error
probability. This contrasts with standard information-theoretic metrics of
security, which are concerned with the average number of bits that an adversary
learns about the plaintext. Nevertheless, as shown here, restrictions on the
average mutual information can be mapped to lower bounds on average estimation
error probability through  rate-distortion formulations. 

Using a rate-distortion
based approach, we extend the definition of
$\epsilon$-symbol secrecy in order to quantify not only the information that an
adversary gains about individual symbols of the source sequence, but also the
information gained about \textit{functions} of the encrypted source sequence. We
prove that ciphers with high symbol secrecy  guarantee that certain
functions of the plaintext are provably hidden regardless of computational
assumptions. In particular, we show that certain one-bit function of the
plaintext (i.e. predicates) cannot be reliably inferred by the adversary.


We illustrate the application of our results both for hiding the source data and
functions thereof. We  provide an extension of the  one-time
pad \cite{shannon_communication_1949}  to a functional setting, demonstrating how certain classes of functions of the
plaintext can be hidden using a short key. We also consider the privacy against statistical
inference setup studied in \cite{inproc:allertonPriv}, and show how the analysis
introduced here sheds light on the fundamental privacy-utility tradeoff. 

 From a  practical standpoint, we investigate the problem of secure content caching and
 distribution. We propose a hybrid encryption scheme based on list-source codes,
 where a large fraction of the message can be encoded and distributed using a
 key-independent list-source code. The information necessary to resolve the
 decoding list, which can be much smaller than the whole message, is then
 encrypted using a secure  method. This scheme allows a significant amount of
 content to be distributed and cached  \textit{before} dealing with key
 generation, distribution and management issues.

\subsection{Related work}

Shannon's seminal work \cite{shannon_communication_1949} introduced the use of
statistical and information-theoretic metrics for analyzing secrecy systems.
Shannon characterized several properties of conditional entropy (equivocation)
as a metric for security, and investigated the effect of the source distribution
on the security of a symmetric-key cipher.  Shannon also considered the
properties of ``random ciphers'', and showed that, for short keys and
sufficiently long, non-uniformly distributed messages, the random cipher is
(with high probability) breakable: only one message is very likely to have
produced a given ciphertext. Shannon defined the length of the message required
for a ciphertext to be uniquely produced by a given plaintext as the
\textit{unicity distance}.

Hellman extended Shannon's approach to cryptography
\cite{hellman_extension_1977} and proved that Shannon's random cipher model is
conservative: A randomly chosen cipher is likely to have small unicity distance,
but does not preclude the existence of other ciphers with essentially infinite
unicity distance (i.e. the plaintext cannot be uniquely determined from the
ciphertext). Indeed, Hellman argued that carefully designed ciphers that match
the statistics of the source can achieve high unicity distance. Ahlswede
\cite{ahlswede_remarks_1982} also extended Shannon's theory of secrecy systems
to the case where the exact source statistics are unknown.

The problem of quantifying not only an eavesdropper's uncertainty of the entire
message but of individual symbols of the message was studied by Lu in the
context of additive-like instantaneous block ciphers (ALIB)
\cite{lu_existence_1979,lu_random_1979,lu_secrecy_1979}. The results presented
here are more general since we do not restrict ourselves to ALIB ciphers. More
recently, the design of secrecy systems with distortion constraints on the
adversary's reconstruction was studied by Schieler and Cuff
\cite{schieler_rate-distortion_2014}. We adopt here an alternative approach,
quantifying the information an adversary gains on average about the individual
symbols of the message, and investigate which functions of the plaintext an
adversary can reconstruct. Our results and definitions also hold for the
finite-blocklength regime.

Tools from algebraic coding  have been widely used for constructing
secrecy schemes  \cite{blahut_communications_1994}. In addition, the notion of
providing security by exploiting the fact that the adversary has incomplete
access to information (in our case, the key) is also central to several secure
network coding schemes and wiretap models. Ozarow and Wyner
\cite{ozarow_wire-tap_1985} introduced the wiretap channel II, where an
adversary can observe a  set $k$  of his choice out of $n$ transmitted symbols,
and proved that there exists a code  that achieves perfect secrecy. A
generalized version of this model was investigated by Cai and Yeung in
\cite{cai_secure_2002}, where they introduce  the related problem of designing
an information-theoretically secure linear network code when an adversary can
observe a certain number of edges in the network. Their results were later
extended in
\cite{feldman_capacity_2004,mills_secure_2008,el_rouayheb_secure_2012,silva_universal_2011}.
A  practical approach was presented by Lima \etal in
\cite{lima_random_2007}. For a survey on the theory of secure network coding, we
refer the reader to \cite{cai_theory_2011}.

The list-source code framework introduced here is related to the wiretap channel
II in that a fraction of the source symbols is hidden from a possible adversary.
Oliveira \etal investigated in \cite{oliveira_trusted_2010} a related setting in
the context of data storage over untrusted networks that do not collude,
introducing a solution based on Vandermonde matrices. The MDS coding scheme
introduced in this paper is similar to \cite{oliveira_trusted_2010}, albeit the
framework developed here  is more general.

List decoding techniques for channel coding were first introduced by Elias
\cite{elias_list_1957} and Wozencraft \cite{wozencraft1958}, with subsequent
work by Shannon \etal \cite{shannon_lower_1967,shannon_lower_1967-1}  and Forney
\cite{forney_exponential_1968}. Later, algorithmic results for list
decoding of channel codes were discovered by Gurusuwami and Sudan
\cite{guruswami_list_2001}. We refer the reader to \cite{guruswami_list_2009}
for a survey of  list decoding results. List decoding has been
considered in the context of source coding in \cite{ali_source_2010}. The
approach is related to the one presented here, since we may view a secret key as
side information, but \cite{ali_source_2010} did not consider source coding and
list decoding together for the purposes of security. 

The use of rate-distortion formulations in security and privacy settings was
studied by Yamamoto \cite{yamamoto_rate-distortion_1997} and Reed \cite{reed_information_1973}. Information-theoretic approaches to privacy that
take distortion into account were also considered in \cite{inproc:allertonPriv,sarwate_rate-disortion_2014,rebollo-monedero_t-closeness-like_2010,sankar_utility-privacy_2013}. 

\subsection{Notation}
\label{sec:notation}

Throughout the paper capital letters (e.g. $X$ and $Y$) are used to denote
random variables, and calligraphic letters (e.g. $\calX$ and $\calY$) denote
sets. All the random variables in this paper have a discrete support set, and
the support set of the random variables $X$ and $Y$ are denoted by $\calX$ and
$\calY$, respectively. For a
positive integer $j,k,n$, $j\leq k$, $[n]\defined  \{1,\dots,n\}$,
$[j,k]\defined \{j,j+1,\dots,k\}$. Matrices are denoted in bold
capital letters (e.g. $\bH$) and vectors in bold lower-case letters (e.g.
$\mathbf{h}$). A sequence of $n$ random variables $X_1,\dots,X_n$ is denoted by
$X^n$. Furthermore, for $\calJ\subseteq [ n]$, $X^{\calJ}\defined\left(X_{i_1},\dots,X_{i_{|\calJ|}}\right)$ where
$i_k\in \calJ$ and $i_1<i_2<\dots<i_{|\calJ|}$. Equivalently, for a vector $\bx =
(x_1,\dots,x_n)$, $\bx^\calJ\defined
\left(x_{i_1},\dots,x_{i_{|\calJ|}}\right)$. For two vectors
$\bx,\bz\in\Reals^n$, we denote by  $\bx\leq\bz$ the set of inequalities $\{x_i\leq
z_i\}_{i=1}^n$. Furthermore, we denote by $\mathcal{I}_n(t)$ the
set of all subsets of $[n]$ of size $t$, i.e.  $\calJ\in \mathcal{I}_n(t)
\Leftrightarrow \calJ\subseteq [n]$ and $|\calJ| = t$.

All the logarithms in the paper are in base 2. We denote the binary entropy function as 
\[h_b(x)\defined-x\log  x-(1-x)\log (1-x).\]
The inverse of the binary entropy function is the mapping
$h_b^{-1}:[0,1]\to[0,1/2]$ where
\begin{align*}
  h_b^{-1}(h(x))=
    \begin{cases}
        x,&0\leq x \leq 1/2\\
        1-x, &\mbox{otherwise.}
    \end{cases}
\end{align*}
The set of all unit variance functions of a random variable $X$ with
distribution $p_X$ (denoted by $X\sim p_X$) is given by 
\begin{equation*}
  \calL_2(p_X) \defined \left\{ \phi:\calX\to \Reals \mbox{~such that~}
  \|\phi(X)\|_2= 1,~X\sim p_X
   \right\},
\end{equation*}
where $\|\phi(X)\|_2\defined \sqrt{\EE{\phi(X)^2}}$.

The operators $T_X$ and $T_Y$  denote conditional expectation and, in
particular, $(T_X\circ g)(x)= \EE{g(Y)|X=x}$ and $(T_Y\circ f)(y)=\EE{f(X)|Y=y}$,
respectively. For two random variables $X$ and $Y$, the minimum-mean-squared
error (MMSE) of estimating $X$ from an observation of $Y$ is given by
\begin{equation*} 
  \mmse(X|Y)\defined\min_{X\rightarrow Y\rightarrow \hat{X}}
  \EE{(X-\hat{X})^2}.  
\end{equation*}

\subsection{Communication and threat model} 
\label{sec:model}

A transmitter (Alice)  wishes to transmit confidentially to a legitimate receiver
(Bob) a sequence of length $n$ produced by a discrete  source $X$ with alphabet
$\mathcal{X}$ and probability distribution $p_{X}$. We assume that the
communication channel shared by Alice and Bob is noiseless, but is observed by a
passive, computationally unbounded eavesdropper (Eve). Both Alice and Bob
have access to a shared secret key $K$ drawn  from a discrete alphabet
$\mathcal{K}$, such that $H(K)< H(X^n)$, and encryption/decryption functions
$\Enc:\calX^n\times \mathcal{K} \rightarrow \mathcal{M}$ and
$\Dec:\mathcal{M}\times \mathcal{K}\rightarrow \calX^n$, where $\mathcal{M}$ is
the set  encrypted messages. Alice observes the source sequence $X^n$, and
transmits an encrypted message $M= \Enc(X^n, K)$. Bob then recovers $X^n$ by
decrypting the message using the key, producing $\hat{X}^n=\Dec(M,K)$. The
communication is successful if $\hat{X}^n=X^n$.  We consider that the encryption
is closed \cite[pg.  665]{shannon_communication_1949}, so $\Dec(c,k_1)\neq
\Dec(c,k_2)$  for $k_1,k_2\in \mathcal{K} $, $k_1\neq k_2$.  We assume Eve knows
the functions $\Enc$ and $\Dec$, but does not know the secret key, $K$. Eve's
goal is to gain knowledge about the original source sequence.  


\subsection{Organization of the paper}

\subsubsection{Symbol secrecy} 
We introduce the definitions of absolute and
$\epsilon$-symbol secrecy in Section \ref{sec:symbolsecrecy}. Symbol secrecy quantifies the uncertainty
that an eavesdropper has about individual symbols of the message.

\subsubsection{Encryption with key entropy smaller than the message entropy} 
We present the definition of list-source codes
(LSCs), together with fundamental bounds, in Section \ref{sec:LSC}.
Practical code constructions of LSCs are  introduced in Section
\ref{sec:code}. We then analyze the symbol secrecy properties of LSCs in Section
\ref{sec:LSC_SS}.

\subsubsection{A Rate-Distortion View of Symbol Secrecy} In Section
\ref{sec:beyond} we introduce
results for characterizing the information leakage of a security system in terms
of functions of the original source data. In particular, we derive converse
bounds for the minimum-mean-squared error (MMSE) of estimating a target function
of the plaintext given that  certain functions of the plaintext are known to
be hard (or easy) to infer. We illustrate the application of these bounds in a
generalization of the one-time pad. We also use these results to bound the
probability of error of estimating predicates of the plaintext given that a
certain level of symbol secrecy is achieved.

\subsubsection{Further applications and practical considerations} 
 Section \ref{sec:practical} presents further applications of our results to
 security and privacy, together with practical considerations of the
proposed secrecy framework. Finally, Section \ref{sec:conc} presents our concluding
remarks.

\section{Symbol Secrecy}
\label{sec:symbolsecrecy}

In this section we define $\epsilon$-symbol secrecy, an information-theoretic
metric for  quantifying the information leakage from  security schemes that do
not achieve perfect secrecy. Given a source sequence $X^n$ and a random variable
$Z$ dependent of $X^n$, $\epsilon$-symbol secrecy is  the largest fraction $t/n$
such that, given $Z$,  at most $ \epsilon$ bits can be learned \textit{on
average} from any $t$-symbol subsequence of $X^n$. We also prove an ancillary
lemma that bounds the average mutual information between $X^n$ and $Z$ in terms
of symbol secrecy.


\begin{defn}
   Let $X^n$ be a random variable with support $\calX^n$, and $Z$ be the information
   that leaks from a security system (e.g. the ciphertext).
   Denoting $X^\calJ=\{X_i\}_{i\in \calJ}$, we say that $p_{X^n,Z}$ achieves an
   $\epsilon$-symbol secrecy of $\mu_\epsilon(X^n|Z)$ if
 \begin{equation}
   \label{eq:def_ssecrecy}
   \mu_\epsilon(X^n|Z)\defined \max\left\{ \frac{t}{n}\,\middle|\,
   \frac{I(X^{\calJ};Z)}{|\calJ|}\leq \epsilon~~\forall \calJ\subseteq [n],
0< |\calJ|\leq t\right\}.
 \end{equation}
 In particular, the \textit{absolute symbol secrecy} of $X^n$ from $Y$ is given by
 \begin{equation}
   \mu_0(X^n|Z)\defined \max\left\{ \frac{t}{n}\,\middle|\,
   I(X^{\calJ};Z)= 0~~ \forall \calJ\subseteq [n],
 0<|\calJ|\leq t\right\}.
 \end{equation}

 We also define the dual function of
symbol-secrecy for $X^n$ and $Z$ as:
    \begin{equation}
      \label{eq:def_estar}
      \epsilon^*_t(X^n|Z)\defined \inf \left\{ \epsilon\geq 0
        \,\middle|\,\mu_{\epsilon}(X^n|Z)\geq t/n \right\}.
    \end{equation}
\end{defn}  

The next examples illustrate a few use cases of symbol secrecy.

\begin{example}
    Symbol secrecy encompasses other definitions of secrecy, such as weak secrecy
    \cite{wyner_wire-tap_1975}, strong secrecy \cite{maurer_information-theoretic_2000} and perfect secrecy. For example, given two
    sequences of random variables $X^n$ and $Z^n$, if
    $\mu_\epsilon(X^n|Z^n)\to 1$ for all $\epsilon>0$, then $\frac{I(X^n;Z^n)}{n}\to
    0$. The converse is not true, as demonstrated in Example \ref{example:weak}
    below. Furthermore, $I(X^n;Z^n)=0$ if and only if $\mu_0(X^n|Z^n)=1$.
    Finally, the reader can verify that $I(X^n;Z^n)\to 0 $ if and only if there
    exists a sequence $\epsilon_n=o(n)$ such that  $\mu_{\epsilon_n}(X^n|Z^n)\to 1$.
\end{example}

\begin{example}
  Consider the case where $\calX=\{0,1\}$, $X^n$ is uniformly drawn from
  $\calX^n$, and $Z$ is the result of sending $X^n$ through a discrete
  memoryless erasure channel with erasure probability $\alpha$. Then, for any
   $\calJ\subseteq [n]$, $\calJ\neq \varnothing$,
  \begin{equation*}
    \frac{I(X^\calJ;Z)}{|\calJ|} = (1-\alpha),
  \end{equation*}
  and, consequently, 
  \begin{equation*}
    \mu_\epsilon(X^n|Z)=
        \begin{cases}
          0,& \mbox{for } 0\leq \epsilon < 1-\alpha, \\
            1,& \epsilon\geq 1-\alpha.
        \end{cases}
  \end{equation*}
\end{example}

\begin{example}
  \label{example:weak}
    Now assume again that $X^n$ is a uniformly distributed sequence of $n$ bits, but
    now $Z=X_1$. This corresponds to the case where one bit of the message is
    always sent in the clear, and all the other bits are hidden. Then, for any
    $\calJ\subseteq[n]$ such that $\{1\}\in \calJ$,
    \begin{equation*}
        I(X^\calJ;Z) = 1,
    \end{equation*}
    and, for $0\leq \epsilon <1$,
    \begin{equation*}
           \mu_\epsilon(X^n|Z)=0. 
    \end{equation*}
    Consequently, a non-trivial symbol-secrecy cannot be achieved for
  $\epsilon<1$. In general, if a symbol $X_i$ is sent in the clear,
  then a non-trivial symbol secrecy cannot be achieved for $\epsilon<H(X_i)$.
  Note that $I(X^n;Z)/n\to 0 $, so weak secrecy is achieved.
\end{example}

\begin{example}
    We now illustrate how symbol secrecy does not necessarily capture the
    information that leaks about functions of $X^n$. We  address this issue
    in more detail in Section \ref{sec:beyond}. Still assuming that $X^n$ is a uniformly distributed sequence of $n$ bits,
    let $Y$ be the parity bit of $X^n$, i.e. $Z=\prod_{i=1}^n (-1)^{X_i}$. Then,
    for any $\calJ \subsetneq [n]$,
    \begin{equation*}
        I(X^\calJ;Z)=0,
    \end{equation*}
    and, for $0\leq \epsilon <1$,
    \begin{equation*}
        \mu_\epsilon(X^n|Z) = \frac{n-1}{n},
    \end{equation*}
    and, for $\epsilon\geq 1$, $\mu_\epsilon(X^n|Z)=1$.
\end{example}

The following lemma provides an upper bound for $I(X^n;Z)$ in terms of
$\mu_\epsilon(X^n|Z)$ when $X^n$ is the output of a discrete memoryless source.
\begin{lem}
  Let $X^n$ be the output of a discrete memoryless source $X$, and $Z$ a noisy
  observation of $X^n$. For any $\epsilon$ such that $0\leq
  \epsilon \leq H(X)$, if $\mu_\epsilon(X^n|Z)=u^*$, then  
\begin{align} 
  \label{eq:mueps_rel}
  \frac{1}{n}I(X^{n} ; Z) \leq H(X) - u^*(H(X)-\epsilon).
\end{align}
\end{lem}
\begin{proof}
  Let  $\mu_{\epsilon}(X^n|Z)=u^*\defined t/n$, $\calJ \in \mathcal{I}_n(t)$  and
  $\bar{\calJ}=[n]\backslash \calJ$. Then
\begin{align*}
  \frac{1}{n}I(X^{n} ; Z)
  &=\frac{1}{n}I(X^\calJ;Z)+\frac{1}{n}I(X^{\bar{\calJ}};Z|X^{\calJ})\\
  &\leq \frac{t}{n}\left(\epsilon+
  \frac{1}{t}I(X^{\bar{\calJ}} ; Z|X^{\calJ})   \right)\\
  &\leq u^*\epsilon +\frac{(n-t)}{n}H(X) \\
  &=  H(X) - u^*(H(X)-\epsilon),
\end{align*}
where the first inequality follows from the definition of symbol secrecy, and
the second inequality follows from the assumption that the source is discrete
and memoryless and, consequently, $I(X^{\bar{\calJ}} ;
Z|X^{\calJ})\leq H(X^{\bar{\calJ}}|X^{\calJ})=(n-t)H(X)$.
\end{proof}

The previous result implies that when $\mu_\epsilon(X^n|Z)$ is large,  only a
small amount of information about $X^n$ can be gained from $Z$ on average.
However, even if $I(X^n;Z)$ is large, as long as $\mu_\epsilon(X^n|Z)$ is
non-zero, the uncertainty about $X^n$ given $Z$ will be spread throughout the
individual symbols of the source sequence. This property is desirable for
symmetric-key encryption and, as we shall show in Section \ref{sec:beyond}, can
be extended to determine which functions of $X^n$ can or cannot be reliably
inferred from $Z$. Furthermore, in Section \ref{sec:LSC_SS} we introduce
explicit constructions for symmetric-key encryption schemes that achieve a
provable level of symbol secrecy using the list-source code framework introduced
next.

\section{LSCs}
\label{sec:LSC}

In this section we present the definition of LSCs and derive
fundamental bounds. We also demonstrate how any symmetric-key encryption scheme can
be mapped to a corresponding list-source code.

\subsection{Definition and Fundamental Limits}

We introduce the definition of list-source codes is given below. 
\begin{defn}
  A $(2^{nR},|\mathcal{X}|^{nL},n)$-LSC $(f_n,g_{n,L})$ consists of an encoding
  function $f_n:\mathcal{X}^n\mapsto\cwd$ and a list-decoding function
  $g_{n,L}:\cwd\mapsto \mathcal{P}(\Xn)\backslash \varnothing$, where
  $\mathcal{P}(\Xn)$ is the power set of $\Xn$ and $|g_{n,L}(w)|=
  |\mathcal{X}|^{nL}~\forall w\in\cwd$. The value $R$ is that \textit{rate} of
  the LSC, $L$ is the \textit{normalized list size}, and $
  |\mathcal{X}|^{nL}$ is the \textit{list size}.
\end{defn}

Note that $0\leq L \leq 1$. From an operational point of view, $L$ is a
parameter that determines the size of the decoded list. For example, $L=0$
corresponds to traditional lossless compression, i.e., each source sequence is
decoded to a unique sequence. Furthermore, $L=1$ represents the trivial case
when the decoded list corresponds to $\Xn$. 
 
 For a given LSC, an error is declared
when a string generated by a source is not contained in the corresponding
decoded list. The average error probability is given by
 \begin{equation}       
   \error(f_n,g_{n,L})\defined\PR(X^n\notin g_{n,L}(f_n(X^n))).
\end{equation}
\begin{defn} 
  For a given discrete memoryless source $X$, the rate list size pair $(R,L)$ is
  said to be \textit{achievable} if for every $\delta>0$, $0<\epsilon<1$ and
  sufficiently large $n$ there exists a sequence of
  $(2^{nR_n},|\mathcal{X}|^{nL_n},n)$-list-source codes
  $\{(f_n,g_{n,L_n})\}_{n=1}^\infty$ such that $R_n< R+\delta$, $|L_n- L|<
  \delta$ and $\error(f_n,g_{n,L_n})\leq \epsilon$. The \textit{rate list
  region} is the closure of all rate list pairs $(R,L)$.
\end{defn}

\begin{defn}
  The \textit{rate list function} $R(L)$ is the infimum of all rates $R$
such that $(R,L)$ is in the rate list region for a given normalized list size
$0\leq L \leq 1$.
\end{defn}

\begin{thm}
  \label{thm:ratelist}
For any discrete memoryless source X, the rate list function is given by
\begin{equation} 
\label{eq:ratelist}
  R(L)= H(X)-L\log|\mathcal{X}|~.
\end{equation}
\end{thm}

\begin{figure}[!tb]
  \begin{center}
    \psfrag{1}[c][c]{$R$}
    \psfrag{2}[b][l][1][90]{$L$}
    \psfrag{3}[r][r]{\small$\displaystyle  \frac{H(X)}{\log|\mathcal{X}|}$}
    \psfrag{4}[c][c]{\small$0$}
    \psfrag{5}[c][c]{\small$0$}
    \psfrag{6}[c][c]{\small$H(X)$}
    \psfrag{7}[l][t]{Achievable}
    \psfrag{8}[c][c]{1}
    \includegraphics[scale=0.42]{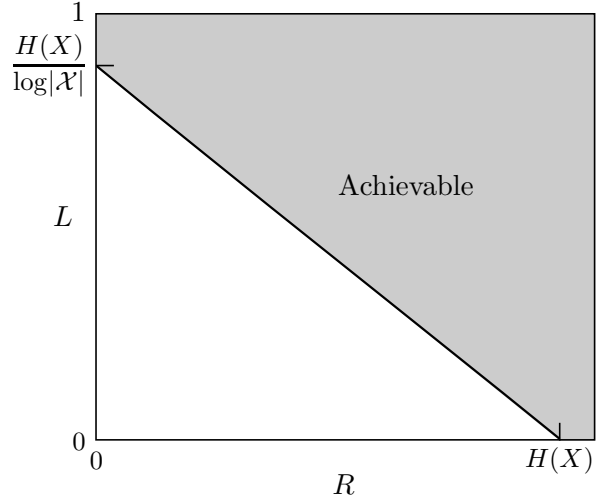}
  \end{center}
  \caption{Rate list region for normalized list size $L$ and code rate $R$.}
  \label{fig:conj}
\end{figure}

\begin{proof}
  Let $\delta>0$ be given and $\{(f_n,g_{n,L_n})\}_{n=1}^\infty$ be a sequence of codes with
(normalized) list size $L_n$ such
that  $L_n\rightarrow L$  and for any $0<\epsilon<1$ and $n$ sufficiently large
$0\leq \error(f_n,g_{n,L_n})\leq \epsilon$. Then 
\begin{align} 
  \PR \left(X^n\in \displaystyle \bigcup_{w\in \mathcal{W}^n} g_{n,L_n}(w) \right)
  &\geq \PR\left(X^n\in g_{n,L_n}(f_n(X^n))\right)\\
  &\geq 1-\epsilon
\end{align}
where $\mathcal{W}^n=[2^{nR_n}]$ and
$R_n$ is the rate of the code $(f_n,g_{n,L_n})$. There exists
$n_0(\delta,\epsilon,|\mathcal{X}|)$ where if  $n\geq
n_0(\delta,\epsilon,|\mathcal{X}|)$, then 
\begin{align}
  R_n+L_n\log |\mathcal{X}|&=\frac{1}{n}
  \log\left(  2^{nR_n}|\mathcal{X}|^{nL_n} \right) \nonumber\\
  &= \frac{1}{n} \log\left( \displaystyle \sum_{w\in \mathcal{W}^n}
  |g_{n,L_n}(w)|\right)\nonumber \\
  &\geq  \frac{1}{n} \log  \left| \bigcup_{w\in
    \mathcal{W}^n} g_{n,L_n}(w)  \right| \nonumber \\ 
  &\geq H(X)-\delta,
\end{align}
where the last inequality follows
from  \cite[Lemma 2.14]{csiszar_information_2011}. Since this holds for any
$\delta>0$, it follows that $R(L)\geq H(X)-L\log |\mathcal{X}|$ for all $n$ sufficiently large.

We prove achievability next. Let $0<L<1$ be given, and let $L_n\defined \lfloor
nL \rfloor$. Furthermore, let $X^n$ be a sequence of $n$ source symbols, and
denote $X^{nL_n}$ the first $nL_n$ source symbols and $X^{[nL_n+1,n]}$ the last
$n(1-L_n)$
source symbols where we assume, without loss of generality, that $nL$ is an
integer. Then, from standard source coding results \cite[pg.
552]{cover_elements_2006}, for any $\epsilon>0$ and $n$ sufficiently large, and
denoting $\alpha_{n} \defined \ceil{nL_n(H(X)+\epsilon)}/n$,  $\beta_{n} \defined
\ceil{n(1-L_n)(H(X)+\epsilon)}/n$, there
are (surjective) encoding functions 
\begin{equation*}
f^1_{nL}:\calX^{nL_n}\to
[2^{n\alpha_n}]
\mbox{ and } f^2_{n(1-L_n)}:\calX^{n(1-L_n)}\to [2^{n\beta_n}],
\end{equation*}
and
corresponding (injective) decoding functions
\begin{equation*}
g^1_{n,1}: [2^{n\alpha_n}]\to
\calX^{nL_n} \mbox{ and } g^2_{n,1}: [2^{n\beta_n}] \to \calX^{nL_n}
\end{equation*}
such that $\Pr(g^1_{n,1}(f^1_{nL_n}(X^{nL_n}))\neq X^{nL_n})\leq O(\epsilon)$ and
$\Pr(g^2_{n,1}( f^2_{n(1-L_n)}(X^{(1-L_n)n}))\neq X^{(1-L_n)n})\leq O(\epsilon)$. 

For $w\in [2^{n\beta_n}]$ and $\bx\in \calX^n$, let the list-source coding
and decoding functions be given by $f_{n}(\bx)\defined f^2_{n(1-L_n)}(\bx^{[nL_n+1,n]})$ and
$$g_{n,\tilde{L}_n}(w)\defined\{\bx\in \calX^n: \exists v\in [2^{n\alpha_n}] \mbox{ such
that } (f^1_{nL}(\bx^{[nL]}),f^2_{n(1-L)}(\bx^{[nL+1,n]}))=(v,w)  \},$$
respectively.  Then
  \begin{align*}
    \Pr\left( X^n\in g_{n,\tilde{L}_n}(f_n(X^n))\right) &\geq \Pr\left(g^1_{n,1}(f^1_{nL}(X^{Ln}))= X^{Ln}\wedge
    g^2_{n,1}(f^2_{n(1-L)}(X^{(1-L)n}))= X^{(1-L)n}\right)\\
    &\geq 1- O(\epsilon).
  \end{align*}

  Observe that the rate-list pair achieved by $(f_{n},g_{n,\tilde{L}_n}
  )$ is $(R_n,\tilde{L}_n)=(\beta_n, \alpha_n/\log|\calX|)) $. Consequently,
\begin{align*}
  R_n &\leq (1-L_n)(H(X)+\epsilon)+n^{-1} \\
  &\leq H(X)+\epsilon-\alpha_n \\
  &=H(X)+\epsilon-\tilde{L}_n\log|\calX|,
\end{align*}
  where the second inequality follows from $\alpha_n \leq
  L_n(H(X)+\epsilon)+n^{-1}$. Observe that $R_n\to n(1-L)H(X) +\epsilon\defined
  R$. Since $\tilde{L}_n\rightarrow
  L(H(X)+\epsilon)/\log|\calX|\defined \tilde{L}$
  as $n\rightarrow \infty$,  by choosing $n$ sufficiently large the rate-list pair $(R,\tilde{L})$ can be
  achieved, where $R$ and $\tilde{L}$ satisfy $$R \leq
  H(X)+\epsilon-\tilde{L}\log|\calX|.$$ Since $\epsilon$ is arbitrary and
  $\tilde{L}$ can span
  any value in $[0,H(X)/\log|\calX|]$, it follows that $R(L)\leq H(X) -
  L\log|\calX|$.
\end{proof}

\subsection{Symmetric-Key Ciphers as LSCs}

Let $(\Enc,\Dec)$ be a symmetric-key cipher where, without loss of generality,
$\calM=[2^{nR}]$ and $\Enc:\calX^n\times \calK\to \calM$ and $\Dec:\calM\times
\calK\to \calX^n$. Then an LSC can be designed based on this cipher by choosing
$k'$ from $\mathcal{K}$ and setting the encoding function
$f_n(\bx)=\Enc(\bx,k')$, where $\bx\in \calX^n$, and \[g_{n,L}(f_n(\bx))=
\{\mathbf{z}\in \calX^n:\exists k\in\mathcal{K} \mbox{ such that }
\Enc(\mathbf{z},k)=f_n(\bx)\},\] where $L$ satisfies $|\calK|=|\calX|^{nL}$. If
the key is chosen uniformly  from $\mathcal{K}$ then the decoded list
corresponds set of possible source sequences that could have generated the
ciphertext. The adversary's uncertainty will depend on the distribution of the
source sequence $X^n$.

Alternatively,  symmetric-key ciphers can also be constructed based on an
$(2^{nR},|\calX|^{nL},n)$-list-source code. Let $(f_n,g_{n,L})$ be the corresponding
encoding/decoding function of the LSC, and assume that the key  is drawn
uniformly from $\mathcal{K}=[|\calX|^{nL}]$, where the normalized list size
$L$ determines the length of the key.  Without loss of generality, we also
assume that  Alice and Bob agree on an ordering of $\calX$ and, consequently,
$\calX^n$ can be ordered using the corresponding dictionary ordering. We denote
$\mathsf{pos}(\bx)$ the position of the source sequence $\bx\in \calX$ in the
corresponding list $g_{n,L}(f_n(\bx))$, where $\mathsf{pos}:\calX^n\rightarrow
[|\calX|^{nL}]$.

The cipher can then be constructed by letting the message set be
$\calM'=[2^{nR}]\times[|\calX|^{nL}]$ and, for $\bx\in \calX^n$ and $k\in
  \mathcal{K}$,
\[\Enc(\bx,k)=(f_n(\bx),(\mathsf{pos}(\bx)+k)\mod|\mathcal{K}|).\] For
$(a,b)\in \calM'$, the decryption function is given by
$$\Dec((a,b ),k)=\{\bx:f_n(\bx)=a,\mathsf{pos}(\bx)=(b-k)\mod|\mathcal{K}|\}.$$ In
this case, an eavesdropper that does not know the key $k$ cannot recover the
function $\mathsf{pos}(\bx)$ and, consequently, her uncertainty
will correspond to the list $g_{n,L}(f_n(\bx))$.

\section{LSC design}
\label{sec:code}

In this section we discuss how to construct LSCs that achieve the
 rate-list tradeoff  \eqref{eq:ratelist} in the finite block length
regime. As shown below, an LSC that achieves good rate-list tradeoff does not
necessarily lead to good symmetric-key encryption schemes. This naturally 
motivates the constructions of LSCs that achieve high symbol secrecy.

\subsection{Necessity for code design}
\label{sec:trivial}
Assume that the source $X$ is
uniformly distributed in $\mathbb{F}_q$, i.e., $\PR(X=x)=1/q~\forall x \in
\mathbb{F}_q$. In this case $R(L)=(1-L)\log q$. A trivial scheme for achieving
the list-source boundary is the following. Consider a source sequence
$X^n=(X^p,X^s)$, where $X^p$ denotes the first $p=n-\floor{Ln}$ symbols of $X^n$
and $X^s$ denotes the last $s=\floor{Ln}$ symbols. Encoding is done by
discarding $X^s$, and mapping the prefix  $X^p$ to a binary codeword $Y^{nR}$ of
length $nR=\ceil{n-\floor{Ln}\log q}$ bits. This encoding procedure is similar
to the achievability scheme used in the proof of Theorem \ref{thm:ratelist}.

For decoding, the codeword $Y^{nR}$ is mapped to $X^p$, and the scheme outputs a
list of size $q^s$ composed by $X^p$ concatenated with all possible combinations
of  suffixes of length $s$. Clearly, for $n$ sufficiently large, $R\approx
(1-L)\log q$, and we achieve the optimal list-source size tradeoff. 

The previous scheme is  inadequate for security purposes. An adversary that
observes the  codeword $Y^{nR}$ can uniquely identify the first $p$
symbols of the source message, and the uncertainty is  concentrated over the
last $s$ symbols. Assuming that all source symbols are of equal
importance, we should  spread the uncertainty over all symbols of the message.
Given the encoding $f(X^n)$, a sensible  security scheme would
provide $I(X_i;f(X^n)) \leq \epsilon\ll \log q$ for $1\leq
i \leq n$. We can naturally extend this notion for groups of symbols
or functions over input symbols, which is  what symbol secrecy captures.

\subsection{A construction based on linear codes}

 Let $X$ be an i.i.d. source with support $ \calX$ and entropy $H(X)$, and
 $(s_n,r_n)$ a   source code for $X$ with  encoder $s_n:\calX^n \rightarrow \Fq^{m_n}$
 and decoder  $r_n:\Fq^{m_n}\rightarrow \calX^n$. Furthermore, let
 $\mathcal{C}$ be a $(m_n,k_n,d)$ linear code\footnote{For an overview of linear
   codes an related terminology, we refer the reader to
   \cite{roth_introduction_2006}.} over $\Fq$ with an
 $(m_n-k_n)\times m_n$ parity check matrix $\mathbf{H}_n$ (i.e. $\bc \in
 \mathcal{C} \Leftrightarrow \bH_n \bc =0 $).  Consider the following scheme,
 where  we assume
 \begin{equation*}
    k_n \defined n L_n \log |\calX|/\log q
  \end{equation*}
 is an integer,  $0\leq L_n \leq 1$ and $L_n \rightarrow L$ as
 $n\rightarrow \infty$. 

\begin{scheme}
\label{scheme:lin}
\textit{Encoding}: Let $\bx_n\in \calX^n$ be an $n$-symbol  sequence generated by the source.  Compute
the syndrome $\bsigma_n$ through the matrix multiplication \[\bsigma_n\defined\bH_n s_n(\bx_n)\] and map  each syndrome to a
distinct  sequence of $nR=\ceil{(m_n-k_n)\log q}$ bits, denoted by
$\by_{nR}$.

\textit{Decoding}: Map the binary codeword  $\by_{nR}$ to the corresponding
syndrome $\bsigma_n$. Output the list
\begin{equation*}
  g_{n,L_n}(\bsigma_n) =\left\{ r_n(\bz)\middle| \bz\in\Fq^{m_n},~ \bsigma_n=\bH_n\bz
\right\}. 
\end{equation*}
\end{scheme}

\begin{thm}
\label{thm:achieve}
If a sequence of source codes $\{(s_n,r_n)\}_{n=1}^\infty$ is asymptotically optimal for source $X$, i.e. $
m_n/n\rightarrow H(X)/\log q$ with vanishing error probability, scheme \ref{scheme:lin} achieves the 
rate list function $R(L)$ for source $X$.
\end{thm}
\begin{proof} 
Since the cardinality of each coset corresponding to a syndrome  $\bsigma_n$ is
exactly \[| g_{n,L_n}(\bsigma_n)|=q^{k_n},\] the normalized list size is
\[L_n=\log_{|\calX|}q^{k_n}=(k_n\log q)/(n\log |\calX|). \] By assumption,
$L_n\to L$ as $n\to \infty$. Denoting $ m_n/n= H(X)/\log
q+\delta_n$, where $\delta_n\rightarrow 0$ since the source code is assumed to
be asymptotically optimal, it follows that the rate of the LSC is   
\begin{align*}
R_n &=\ceil{(m_n-k_n)\log q}/n\\
&=\ceil{(H(X)+\delta_n\log q)n-L_n n \log|\calX|}/n\\
      &\to H(X)-L\log |\calX|,
\end{align*}
which is arbitrarily close to the rate in \eqref{eq:ratelist} for sufficiently
large $n$.
\end{proof}


The source coding scheme used in the proof of Theorem \ref{thm:achieve} can be
any asymptotically optimal scheme. Note that if the source $X$ is uniformly
distributed in $\Fq$, then  $L_n=k_n/n$ and
any message in the coset  indexed by $\bsigma_n$  is equally
likely. Hence, $R_n=(n-k)\log q/n=H(X)-L\log q$, which matches the upper bound
in  \eqref{eq:ratelist}. Scheme \ref{scheme:lin} provides a
constructive way of hiding information, and we can take advantage of the
properties of the underlying linear code to make precise assertions regarding
the security of the scheme.

With the syndrome in hand, how can we recover the rest of the message? One
possible approach is to find a $k_n \times n$ matrix $\bD_n$ that has full rank
such that the rows of $\bD_n$ and $\bH_n$ form a basis of $\Fq^{m_n}$. Such a
matrix can be easily found, for example, using the Gram-Schmidt process with the
rows of $\bH_n$ as a starting point. Then, for a source sequence $\bx_n$, we
simply calculate $\mathbf{t}_n=\bD_n \bx_n$ and forward $\mathbf{t}_n$ to the
receiver through a secure channel. The receiver can then invert the system
\begin{equation}
\left( 
\begin{array}{c}
\bH_n\\
\bD_n
\end{array}
 \right)
 \bx_n=\\
 \left(\begin{array}{c}
\bsigma_n\\
\mathbf{t}_n
\end{array}\right),
\end{equation}
and recover the original sequence $\bx_n$. This property allows list-source codes
to be deployed in practice using well known linear code constructions, such as
Reed-Solomon \cite[Chap. 5]{roth_introduction_2006} or Random Linear Network
Codes \cite[Chap. 2]{ho_network_2008}.

\begin{remark}
This approach is valid for general linear spaces, and holds for any pair of full
rank matrices $\bH_n$ and $\bD_n$ with dimensions $(m_n-k_n)\times m_n$ and
$k_n\times m_n$,
respectively, such that $\rank ([\bH_n^T~\bD_n^T]^T)=m_n$. However, here we adopt the
nomenclature of linear codes since we  make use of known code constructions to
construct LSCs with provable symbol secrecy properties in the next section.
\end{remark}

\begin{remark}
  The LSC described in scheme \ref{scheme:lin} can be combined with other
  encryption methods, providing, for example, an additional layer of security in
   probabilistic encryption schemes
  (\cite{katz_introduction_2007,goldwasser_probabilistic_1984}). A more detailed
  discussion of practical applications is presented in Section
  \ref{sec:practical}.
\end{remark}



\section{Symbol Secrecy of LSCs}
\label{sec:LSC_SS}
We  next present fundamental bounds for the amount of symbol secrecy achievable by
any LSC considering a discrete memoryless source. Since any encryption scheme can
be cast as an LSC, these results quantify the amount of symbol secrecy
achievable 
by any symmetric-key encryption scheme that encrypts a discrete memoryless
source.

\begin{lem}
  Let $\{(f_n,g_n) \}_{n=1}^\infty$ be a sequence of list-source codes that achieves a rate-list pair
$(R,L)$ and an $\epsilon$-symbol secrecy of
$\mu_{\epsilon}\left(X^n|Y^{nR_n}\right)\to\mu_\epsilon$ as $n\to \infty$. Then $0\leq \mu_\epsilon \leq
\min \left\{ \frac{L  \log |\calX|}{H(X)-\epsilon},1\right\}$.
\label{prop:epsilon_bound}
\end{lem}
\begin{proof}
  We denote $\mu_\epsilon(X^n|Y^{nR})=\mu_{\epsilon,n} $. Note that, for
  $\calJ\subseteq [n]$ and $|\calJ|=n\mu_{\epsilon,n}$,
\begin{align*} 
  I(X^{\calJ};Y^{nR_n})&=H(X^{\calJ})-H(X^{\calJ}|Y^{nR_n})\\
  &= n\mu_{\epsilon,n}H(X)-H(X^{\calJ}|Y^{nR_n})\\
  &\leq n\mu_{\epsilon,n} \epsilon,
\end{align*}
where the last inequality follows from the definition of symbol secrecy and
$I(X^\calJ;Y^{nR_n})\leq |\calJ|\epsilon = n\mu_{\epsilon,n}\epsilon$. Therefore
\begin{align*} 
  \mu_{\epsilon,n}(H(X)-\epsilon)&\leq \frac{1}{n}H(X^{\calJ}|Y^{nR_n})\\
  &\leq L_n\log|\mathcal{X}|.
\end{align*}
The result follows by taking $n\rightarrow \infty$.
\end{proof}

The previous result bounds the amount of information an adversary gains about
particular source symbols by observing a list-source encoded
message. In particular, for $\epsilon=0$, we find a meaningful bound on what is
the largest fraction of input symbols that is
\textit{perfectly} hidden. 


The next theorem relates the rate-list function with $\epsilon$-symbol
secrecy through the upper bound in Theorem \ref{prop:epsilon_bound}.

\begin{thm}
  If a sequence of list-source codes $\{(f_n,g_{n,L_n}) \}_{n=1}^\infty$  achieves a
  point $(R',L)$ with $\mu_\epsilon(X^n|Y^{nR_n})\to \frac{L  \log
  |\calX|}{H(X)-\epsilon}\defined c_\epsilon$ for
some $\epsilon$, where $R'=\lim_{n\rightarrow \infty}\frac{1}{n}H(Y^{nR_n})$,
then $R'=R(L)$.
\end{thm}
\begin{proof}
  Assume that  $\{(f_n,g_{n,L_n}) \}_{n=1}^\infty$  satisfies the conditions in the theorem and $\delta>0$
is given. Then for $n$ sufficiently large, we have from \eqref{eq:mueps_rel}:
  \begin{align*}  
    \frac{1}{n}H(Y^{nR_n}) &= \frac{1}{n}I(X^{n} ; Y^{nR_n})\\
    &\leq   H(X) - c_{\epsilon}(H(X)-\epsilon)+\delta\\
    & = H(X) - L\log|\calX|+\delta .
  \end{align*}
Since this holds for any $\delta$, then $R'\leq H(X) - L\log|\cal X|$. However,
from Theorem \ref{thm:ratelist}, $R'\geq H(X) - L\log|\cal X|$, and the result follows. 
\end{proof}

\subsection{A scheme based on MDS codes}

We now prove that for a uniform i.i.d. source $X$ in $\Fq$, using scheme
\ref{scheme:lin} with an MDS parity check matrix $\bH$ achieves $\mu_0$. Since
the source is uniform and i.i.d., no source coding is used.

\begin{prop}
  \label{prop:MDS}
If $\bH$ is the parity check matrix of an $(n,k,d)$ MDS code and the source $X^n$ is
uniform and i.i.d., then Scheme \ref{scheme:lin} achieves the upper bound $\mu_0
= L$, where $L=k/n$.
\end{prop}
\begin{proof}
Let $\calC$ be the set of codewords of an  $(n,k,n-k+1)$ MDS code
over $\Fq$ with parity matrix $\bH$, and let $\bx\in \calC$. Fix a set $\calJ \in
\mathcal{I}_{n}(k)$ of $k$ positions of $\bx$, denoted $\bx^{\calJ }$. Since
the minimum distance of $\calC$ is $n-k+1$, for any other codeword in $\bz \in
\calC$ we have  $\bz^{\calJ}\neq \bx^{\calJ}$. Denoting by
$\calC^{\calJ}=\{\bx^{\calJ} \in \Fq^k:x\in \calC \}$, then
$|\calC^{\calJ}|=|\calC|=q^k.$ Therefore, $\calC^{\calJ}$ contains all
possible combinations of $k$ symbols. Since this property also holds for any
coset of $\bH$, the result follows.

\end{proof}

\section{A Rate-Distortion View of Symbol Secrecy}
\label{sec:beyond}

Symbol secrecy provides a fine-grained metric for quantifying the amount of
information that leaks from a security system. However, standard cryptographic
definitions of security are concerned  not only with what an eavesdropper learns
about individual symbols of the plaintext, but also which \textit{functions} of
the plaintext an adversary can reliably infer.  In order to derive analogous
information-theoretic metrics for security, in this section we take a step back
from the symmetric-key encryption setup and study the general estimation problem of
inferring properties of a hidden variable $X$ from an observation $Y$. More
specifically, we derive lower bounds for the error of estimating functions of $X$ from an
observation of $Y$. By using standard converse results (e.g. Fano's inequality
\cite[Chap. 2]{cover_elements_2006}),
symbol secrecy guarantees are then translated to guarantees on how well certain
functions of the plaintext can or cannot be estimated.

We first derive converse bounds for the minimum-mean-squared-error (MMSE) of
estimating a function $\phi$ of the hidden variable $X$ given $Y$. We assume that
the MMSE of estimating a set of functions $\Phi\defined \{\phi_j(X)\}_{i=1}^m$ given $Y$ is
known, as well as the correlation between $\phi_j(X)$ and $\phi(X)$. Bounds for the
MMSE of $\phi(X)$ are then expressed in terms of the MMSE of each $\phi_j(X)$ and the
correlation between $\phi(X)$ and $\phi_j(X)$.  We also apply this result to the
setting where $\phi$ and $\phi_j$ are binary functions, and present bounds for the
probability of correctly guessing $\phi(X)$ given $Y$. These results are of
independent interest, and are particularly useful in the security setting
considered here.

The set of functions $\Phi$ can be used  to model known properties of a
security system. For example, when $X$ is a plaintext and $Y$ is a ciphertext,
the functions $\phi_j$ may represent certain predicates of $X$ that are known to be
hard to infer given $Y$. In privacy systems,  $X$ may be a user's data and $Y$ a
distorted version of $X$ generated by a privacy preserving mechanism.  The set
$\Phi$ could then represent a set of functions that are known to
be easy to infer from $Y$ due to inherent utility constraints of the setup.  In
particular, as will be shown in Section \ref{sec:from_symbol},
we will consider the functions in $\Phi$ as the individual symbols of the plaintext. In this case, the results
introduced in this section are used to  derive bounds on the MMSE of
reconstructing a target function of the plaintext in terms of the symbol-secrecy
achieved by the underlying list-source code given by the encryption scheme.
This result extends symbol secrecy to a broader setting.

\subsection{Lower Bounds for MMSE}
\label{sec:MMSE}
The results introduced in this section are based on the following Lemma.

\begin{lem}
  \label{lem:quadBound}
  Let $z_n:(0,\infty)^n\times[0,1]^n\rightarrow \Reals$ be given by 
    \begin{equation}
      z_n(\ba,\bb) \defined \max\left\{ \ba^T\by \middle| 
      \by\in\Reals^n,\|y\|_2\leq 1, \by\leq \bb  \right\}. 
      \label{eq:defn_Ln}
   \end{equation}
   
 Let $\pi$ be a permutation of $(1,2,\dots,n)$ such that
       $b_{\pi(1)}/a_{\pi(1)}\leq \dots\leq b_{\pi(n)}/a_{\pi(n)}$.
    If $b_{\pi(1)}/a_{\pi(1)}\geq 1$, $z_n(\ba,\bb)=\|\ba\|_2$. Otherwise,
        \begin{align}
          z_n(\ba,\bb) =&     \sum_{i=1}^{k^*} a_{\pi(i)}b_{\pi(i)} \nonumber \\
          &+ \sqrt{\left(
            \|\ba\|_2^2-\sum_{i=1}^{k^*} a_{\pi(i)}^2
            \right)\left(1-\sum_{i=1}^{k^*}
            b_{\pi(i)}^2\right)}
          \end{align}
     where
     \begin{equation}
       \label{eq:kstar}
        k^*\defined \max\left\{k\in [n] ~\middle|~  
          \frac{b_{\pi(k)}}{a_{\pi(k)}}\leq \sqrt{\frac{\left(1-\sum_{i=1}^{k-1}
          b_{\pi(i)}^2\right)^+}{\|\ba\|_2^2-\sum_{i=1}^{k-1} a_{\pi(i)}^2
        }} \right\}.
       \end{equation}

\end{lem}
\begin{proof}
The proof is given in the appendix.
\end{proof}

Throughout this section we assume   $\Phi\subseteq \calL_2(p_X)$ and
$\EE{\phi_i(X)\phi_j(X)}=0$ for $i\neq j$.
Furthermore, let $Y$ be an observed variable that is dependent of $X$, and for a
given $\phi_i$ the inequality 
\begin{equation*}
  \max_{\psi\in\calL_2(p_Y)} \EE{\phi_i(X)\psi(Y)}=\|\EE{\phi_i(X)|Y} \|_2\leq \lambda_i
\end{equation*}
is satisfied, where
$0\leq \lambda_i\leq 1$. This is equivalent to $\mmse(\phi_i(X)|Y)\geq
1-\lambda_i^2$.

\begin{thm}
  \label{prop:loose}
  Let   $|\EE{\phi(X)\phi_i(X)}|=\rho_i>0$. Denoting
  $\brho\defined (|\rho_1|,\dots,|\rho_m|)$,  $\blambda\defined (\lambda_1,\dots,\lambda_m)$,
  $\rho_0\defined\sqrt{1-\sum_{i=1}^k\rho_i^2} $, $\lambda_0=1$  $\brho_0\defined (\rho_0,\brho)$ and
  $\blambda_0\defined(\lambda_0,\blambda)$, then
  \begin{equation}
    \|\EE{\phi(X)|Y} \|_2\leq B_{|\Phi|}(\brho_0,\blambda_0),
    \label{eq:firstBound}
  \end{equation}
  where
  \begin{equation}
    \label{eq:Bmdef}
    B_{|\Phi|}(\brho_0,\blambda_0)\defined
     \begin{cases}
      z_{|\Phi|+1}\left( \brho_0,\blambda_0 \right),  & \mbox{if } \rho_0>0,\\
    z_{|\Phi|}(\brho,\blambda), & \mbox{otherwise.}  
       \end{cases}
  \end{equation}
  and $z_{n}$ is given in \eqref{eq:defn_Ln}. Consequently,
  \begin{equation}
    \mmse(\phi(X)|Y)\geq 1-B_{|\Phi|}(\brho_0,\blambda_0)^2.
  \end{equation}
\end{thm}
\begin{proof}
  Let $h(X)\defined \rho_0^{-1}(\phi(X)-\sum_i \rho_i \phi_i(X))$ if $\rho_0>0$,
  otherwise $h(X)=0$. Note
  that $h(X)\in \calL_2(p_X)$. Then
  for $\psi\in \calL_2(p_Y)$
  \begin{align*}
    \left| \EE{\phi(X)\psi(Y)} \right| &= \left| \rho_0\EE{h(X)\psi(Y)}+\sum_{i=1}^m \rho_i
    \EE{\phi_i(X)\psi(Y)} \right| \\
    &\leq  \rho_0 \left| \EE{h(X)\psi(Y)}\right|+\sum_{i=1}^m  \left|\rho_i
\EE{\phi_i(X)\psi(Y)} \right| \\
&= \rho_0 \left| \EE{h(X)(T_X\psi)(X)}\right|+\sum_{i=1}^m  \left|\rho_i
\EE{\phi_i(X)(T_X\psi)(X)} \right|.
  \end{align*}
    Denoting $|\EE{h(X)(T_X\psi)(X)}|\defined x_0$, $ |\EE{\phi_i(X)(T_X\psi)(X)}|\defined x_i$,
    $\bx \defined (x_0,x_1,\dots,x_m)$, and $\brho \defined
    (\rho_0,|\rho_1|,\dots,|\rho_m|)$,  the last inequality can be rewritten as
    \begin{align}
          \label{eq:bound_x}
         \left| \EE{\phi(X)\psi(Y)} \right| &\leq \brho_0^T \bx.
    \end{align}    
    
   Observe that $\|\bx\|_2 \leq 1$ and $x_i\leq \lambda_i$ for $i=0,\dots,m$,
   and the right hand side of \eqref{eq:bound_x} can be maximized over all
   values of $\bx$ that satisfy these constraints. We assume, without loss of
   generality, that $\rho_0> 0$ (otherwise  set
   $x_0=0$). The left-hand side of \eqref{eq:bound_x}  can be
   further bounded by 
    \begin{equation}
      \left| \EE{\phi(X)\psi(Y)} \right| \leq z_{m+1}(\brho_0,\blambda_0),
    \end{equation}
    where $\blambda=(1,\lambda_1,\dots,\lambda_m)$ and $z_{m+1}$ is defined in
    \eqref{eq:defn_Ln}. The result follows directly from Lemma
    \ref{lem:quadBound} and noting that $\max_{\psi\in \calL_2(p_Y)}\EE{\phi(X)\psi(Y)}=\|
    \EE{\phi(X)|Y}\|_2.$    
\end{proof}

Denote $\psi_i\defined T_Y \phi_i/\|T_Y \phi_i\|_2$ and $\phi_0(X) \defined (\phi(X)-\sum_{i=1}^m\rho_i
\phi_i(X))/\rho_0^{-1}$. The previous bound can be further improved when
$\EE{\psi _i(Y)\phi_j(X)}=0$ for $i\neq j,\, j\in \{0,\dots,m\}$.

\begin{thm}
  \label{prop:tighter}
  Let  $|\EE{\phi(X)\phi_i(X)}|=\rho_i>0$ for $\phi_i\in \Phi$. In
  addition, assume $\EE{\psi _i(Y)\psi _j(Y)}=0$ for $i\neq j$, $i\in [t]$ and $j\in
  \{0,\dots,|\Phi|\}$, where $0\leq t\leq |\Phi|$. Then
  \begin{equation}
    \|\EE{\phi (X)|Y} \|_2\leq \sqrt{\sum_{k=1}^t \lambda_i^2\rho_i^2 +
    B_{|\Phi|-t}\left(\tilde{\brho},\tilde{\blambda}\right)^2},
  \end{equation}
    where $\tilde{\brho}=(\rho_0,\rho_t,\dots,\rho_m)$,
    $\tilde{\blambda}=(1,\lambda_t,\dots,\lambda_m)$ and $B_m$ is defined in
    \eqref{eq:Bmdef} (considering $B_0=0$). In particular, if $t=m$,
    \begin{equation} 
      \|\EE{\phi (X)|Y} \|_2\leq  \sqrt{\rho_0^2+\sum_{k=1}^{|\Phi|} \lambda_i^2\rho_i^2 },
      \label{eq:sharp}
    \end{equation}
    and this bound is tight when $\rho_0=0$. Furthermore,
    \begin{equation} 
        \mmse(\phi (X)|Y)\geq 1-\sum_{k=1}^t \lambda_i^2\rho_i^2-
    B_{|\Phi|-t}\left(\tilde{\brho},\tilde{\blambda}\right)^2.
    \end{equation}
\end{thm}
\begin{proof}
  For any $\psi \in \calL_2(p_Y)$, let $\alpha_i\defined \EE{\psi (Y)\psi _i(Y)}$ and
  $\psi _0(Y)\defined (\psi (Y)-\sum_{i=1}^t \alpha_i \psi _i(Y))\alpha_0^{-1}$, where
  $\alpha_0=(1-\sum_{i=i}^t\alpha_i^2)^{-1/2}$. Observe that $\psi _0\in \calL_2(p_Y)$
  and $\EE{\phi _i(X)\psi _j(Y)}=\EE{\psi _i(Y)\psi _j(Y)}=0$ for $i\neq j$,
  $i\in \{0,\dots,|\Phi|\}$ and $j\in[t]$.
  Consequently
  \begin{align}
    \EE{\phi (X)\psi (Y)}   &= \EE{\left(\sum_{i=0}^{|\Phi|} \rho_i\phi _i(X)  \right)\left(
    \sum_{j=0}^t \alpha_j \psi _j(Y)\right)} \nonumber\\
    &= \sum_{i=0}^{|\Phi|}\sum_{j=0}^t \rho_i\alpha_j \EE{\phi _i(X)\psi _j(Y)} \nonumber\\ 
    &\leq \left| \alpha_0 \sum_{i=0,i\notin [n] }^{|\Phi|} \rho_i\EE{\phi _i(X)\psi _0(Y)}\right|+\sum_{i=1}^t
    |\lambda_i\rho_i\alpha_i| \nonumber\\
    &\leq  |\alpha_0| B_{|\Phi|-t}\left(\tilde{\brho},\tilde{\blambda}\right)+\sum_{i=1}^t
    |\lambda_i\rho_i\alpha_i| \label{eq:prooftight1}\\
    &\leq \sqrt{\sum_{i=1}^t \lambda_i^2\rho_i^2 +
    B_{|\Phi|-t}\left(\tilde{\brho},\tilde{\blambda}\right)^2} \label{eq:prooftight2}.
   \end{align}
Inequality \eqref{eq:prooftight1} follows from the bound \eqref{eq:firstBound},
and   \eqref{eq:prooftight2} follows by observing that  $\sum_{i=0}^t
\alpha_i^2=1$ and applying the Cauchy-Schwarz inequality.

Finally, when $\rho_0=0$, \eqref{eq:prooftight2} can be achieved with equality
by taking $\psi =\sum_i \frac{\lambda_i\rho_i}{\sqrt{\sum_i\lambda_i^2\rho_i^2}} \psi _i$.

\end{proof}

The following three, diverse examples illustrate different usage cases of
Theorems \ref{prop:loose} and \ref{prop:tighter}.  Example
\ref{example:noise} illustrates Theorem \ref{prop:tighter} for the binary
symmetric channel. In this case, the basis $\Phi$ can be conveniently expressed
as the parity bits of the input to the channel.
 Example \ref{example:q}
illustrates how Theorem   \ref{prop:tighter} can be applied to the $q$-ary
symmetric channel, and demonstrates that bound \eqref{eq:sharp}
 is sharp.  Finally, Example \ref{example:equal}
then illustrates Theorem  \ref{prop:loose} for the specific case where all the
values $\rho_i$ and $\lambda_i$ are equal.

\begin{example}[Binary Symmetric Channel]
  \label{example:noise}
  Let $\calX = \{-1,1\}$ and $\calY=\{-1,1\}$, and $Y^n$ be the result of
  passing $X^n$ through a memoryless binary symmetric channel with crossover
  probability $\epsilon$. We also assume that $X^n$ is composed by $n$ uniform
  and i.i.d. bits. For $\calS\subseteq [n]$, let $\chi_\calS(X^n)\defined\prod_{i\in \calS}
  X_i$. Any function $\phi :\calX\rightarrow \Reals$ can then be decomposed in terms
  of the basis of functions $\chi_\calS(X^n)$ as \cite{odonnell_topics_2008}
  \begin{equation*}
        \phi (X^n)=\sum_{\calS\subseteq[n]}c_\calS \chi_\calS(X^n),
  \end{equation*}   
  where $c_\calS=\EE{\phi (X^n)\chi_\calS(X^n)}$. Furthermore, since
  $\EE{\chi_\calS(X^n)|Y^n}=(1-2\epsilon)^{|\calS|}$, it follows from
  Theorem \ref{prop:tighter} that
  \begin{equation}
    \mmse(\phi (X^n)|Y^n) =1-\sum_{\calS\subseteq[n]}
    c_\calS^2(1-2\epsilon)^{2|\calS|}.
  \end{equation} 
  This result can be generalized for the case where $X^n=Y^n\otimes Z^n$, where
  the operation $\otimes$ denotes bit-wise multiplication, $Z^n$ is drawn from
  $\{-1,1\}^n$ and $X^n$ is uniformly distributed. In this case
  \begin{equation}
        \mmse(\phi (X^n)|Y^n) =1-\sum_{\calS\subseteq[n]}
        c_\calS^2\EE{\chi_\calS(Z^n)}^2.
  \end{equation}
  This example will be revisited in Section \ref{sec:crypto}, where we restrict
  $\phi$
  to be a binary function.
\end{example}

\begin{example}[$q$-ary symmetric channel]
  \label{example:q}
  For $\calX=\calY=[q]$, an $(\epsilon,q)$-ary symmetric channel is defined by the
  transition probability
  \begin{equation}
    p_{Y|X}(y|x)= (1-\epsilon)\indicator_{y=x} +\epsilon/q.
  \end{equation}
  Any function $\phi _i\in\calL_2(p_X)$ such that $\EE{\phi _i(X)}=0$ satisfies
    \begin{align*}
      \psi _i(Y)=T_Y\phi (X)=(1-\epsilon)\phi (Y),
    \end{align*}
    and, consequently, $\|T_Y\phi (X)\|_2 = (1-\epsilon)$. We shall use this fact to
    show that the bound \eqref{eq:sharp} is sharp in this case.
    
    Observe that for  $\phi _i,\phi _j\in \calL_2(p_X)$, if $\EE{\phi _i(X)\phi _j(X)}=0$ then
    $\EE{\psi _i(Y)\psi _j(Y)}=0$. Now let $\phi \in \calL_2(p_X)$ satisfy
    $\EE{\phi (X)}=0$ and $\EE{\phi (X)\phi _i(X)}=\rho_i$ for $\phi_i\in \Phi$,
    where $|\Phi|=m$, $\Phi$
    satisfies the conditions in Theorem \ref{prop:tighter}, and $\sum_i
    \rho_i^2=1$. In addition, $\|\psi _i\|_2=(1-\epsilon)=\lambda_i$. Then, from \eqref{eq:sharp},
    \begin{align*}
      \|T_Y\phi (X)\|_2 &\leq \sqrt{\sum_{i=1}^m \lambda_i^2\rho_i^2}\\
                    &= (1-\epsilon)\sqrt{\sum_i\rho_i^2}\\
                    &= 1-\epsilon, 
      \end{align*}
   which matches $\|T_Y\phi (X)\|_2$, and the bound is tight in this case.
\end{example}

\begin{example}[Equal MMSE and correlation]
  \label{example:equal}
  We now turn our attention to Theorem  \ref{prop:loose}. Consider the case when
the correlations of $\phi$ with the references functions $\phi_i$ are all the
same, and each $\phi_i$ can be estimated with the same MMSE, i.e.
$\lambda_1=\dots=\lambda_m=\lambda$ and $\rho_1^2=\dots=
  \rho_m^2=\rho^2$, $\rho\geq0$ and $\lambda^2\leq \rho^2\leq 1/m$. Then bound
  \eqref{eq:firstBound} becomes
        \begin{equation*}
          \|\EE{\phi (X)|Y} \|_2 \leq m\lambda\rho
          +\sqrt{(1-m\rho^2)(1-m\lambda^2)}.
        \end{equation*}
\end{example}

\subsection{One-Bit Functions}

Let $X$ be a hidden random variable and  $Y$ be a noisy
observation of $X$. Here we denote $\Phi=\{\phi_i \}_{i=1}^m$ a collection of $m$ predicates
of $X$, where $F_i = \phi _i(X)$,    $\phi _i:\calX\rightarrow \{-1,1\}$ for $i\in
[m]$ and, without loss of generality $\EE{F_i}=b_i\geq 0$. 

We denote by $\hat{F}_i$  an estimate of $F_i$ given an observation of $Y$,
where $F_i\rightarrow X \rightarrow Y \rightarrow \hat{F}_i$. We assume that for
any $\hat{F}_i$  $$\left|\mathbb{E}[F_i\hat{F}_i]\right|\leq 1-2\alpha_i $$ for
some $0\leq \alpha_i\leq (1-b_i)/2 \leq 1/2$. This condition is equivalent to
imposing that $\Pr\{F_i\neq \hat{F}_i\}\geq \alpha_i$, since
\begin{align*}
    \EE{F_i\hat{F}_i}&=\Pr\{F_i=\hat{F}_i\}-\Pr\{F_i\neq
      \hat{F}_i\}\\
      &=1-2\Pr\{F_i\neq \hat{F}_i\}.
\end{align*}
In particular, this captures how well $F_i$ can be guessed based solely on
an observation of $Y$.

Now assume there is a bit $F=\phi(Y)$ such that $\EE{F F_i}=\rho_i$ for $i\in [m]$ and
$\EE{F_iF_j}=0$ for $i\neq j$. We
can apply the same method used in the proof of Theorem \ref{prop:loose} to
bound the probability of
$F$ being  guessed correctly from an observation of $Y$.

\begin{cor}
  \label{prop:onebit}
     For $\lambda_i = 1-2\alpha_i$,
     \begin{equation}
       \Pr(F\neq \hat{F}) \geq \frac{1}{2}\left( 1-B_{|\Phi|}(\brho,\blambda)
       \right).
     \end{equation}
\end{cor}

\begin{proof}
  The proof follows the same steps as Theorem \ref{prop:loose},
  $\phi(Y)\in \calL_2(p_Y)$.
\end{proof}

In the case $m=1$, we obtain the following simpler bound, presented in
Proposition \ref{prop:simple}, which depends on the following Lemma.

\begin{lem}
  \label{lem:triangle}
    For any random variables $A,B$ and $C$
    \begin{align*}
        \Pr(A\neq B) \leq \Pr(A\neq C) + \Pr(B\neq C).
    \end{align*}
\end{lem}
\begin{proof}
  \begin{align*}
    \Pr(A\neq B) &= \Pr(A\neq B \land B=C)+\Pr(A\neq B \land B\neq C)\\
                 &= \Pr(A\neq C \land B=C)+\Pr( B\neq C)\Pr(A\neq B | B\neq C)\\
                 &\leq \Pr(A\neq C)+ \Pr( B\neq C).
  \end{align*}
\end{proof}

\begin{prop}
  \label{prop:simple}
  If $\Pr(F_1\neq \hat{F}_1)\geq \alpha$ for all $\hat{F}_1$ and
  $\EE{FF_1}=\rho\geq 0$. Then for any estimator $\hat{F}$
  \begin{equation}
    \Pr(F\neq \hat{F})\geq \left(\frac{1-\rho}{2}-\alpha\right)^+.
  \end{equation}
\end{prop}
\begin{proof}
    From Lemma \ref{lem:triangle}:
    \begin{align*}
        \Pr(F\neq \hat{F})&\geq \left(\Pr(F_1\neq F) -\Pr(F_1\neq \hat{F})\right)^+\\
                            &\geq \left(\frac{1-\rho}{2}-\alpha\right)^+.
    \end{align*}   
\end{proof}

\subsection{ One-Time Pad Encryption of Functions with Boolean Inputs}

\label{sec:crypto}
We  return to the setting where a legitimate transmitter
(Alice) wishes to communicate a plaintext message $X^n$ to a legitimate receiver
(Bob)  through a channel observed by an eavesdropper (Eve). Both Alice and Bob
share a secret key $K$ that is not known by Eve. Alice and Bob use a symmetric
key encryption scheme determined by the pair of encryption and decryption functions
$(\Enc,\Dec)$, where $Y^n=\Enc(X^n,K)$ and $X^n=\Dec(Y^n,K)$. Here we assume that
both the ciphertext and the plaintext have the same length.


We use the results derived in the previous section to assess the security
properties of the one-time pad with non-uniform key distribution when no
assumptions are made on the computational resources available to Eve. In this
case, perfect secrecy (i.e. $I(X^n;Y^n)=0$) can only be achieved when $H(K)\geq
H(X^n)$ \cite{shannon_communication_1949}, which, in turn, is challenging in practice. Nevertheless, as we shall
show in this section, information-theoretic security claims can still be made in the short
key regime, i.e. $H(K)<H(X^n)$. We first prove the following ancillary result.

\begin{lem}
  \label{lem:erasure}
    Let $F$ be a  Boolean random variable and $F\rightarrow
    X\rightarrow Y\rightarrow \hat{F}$, where 
    $|\calY|\geq 2$.
    Furthermore, $\Pr\{F\neq \hat{F}\}\geq \alpha$ for all $Y\rightarrow
    \hat{F}$. Then $I(F;Y)\leq 1-2\alpha$.
\end{lem}
\begin{proof}
    The result is a direct consequence of the fact that the  channel with binary
    input and finite output alphabet that maximizes mutual information for a
    fixed error probability is the erasure channel, proved next. Assume, without
    loss of generality, that $\calY=[m]$ and $p_{F,Y}(-1,y)\geq p_{F,Y}(1,y)$
    for $y\in [k]$ and $p_{F,Y}(-1,y)\leq p_{F,Y}(1,y)$ for $y\in
    \{k+1,\dots,m\}$, where $k\in [m]$. Now let $\tilde{Y}$ be a random variable
    that takes values in $[2m]$ such that
    \begin{align*}
      p_{ F ,\tilde{Y}}(b,y)=
      \begin{cases}
        p_{F ,Y}(b,y)-p_{F ,Y}(1,y) & y\in [k],\\
        p_{F ,Y}(b,y)-p_{F ,Y}(-1,y) & y\in \{k+1,\dots,m\},\\
        p_{F ,Y}(1,y) & y-m \in [k], \\
        p_{F ,Y}(-1,y) & y-m \in \{k+1,\dots,m\}.
      \end{cases}
    \end{align*}.
   
    Note that $F \rightarrow \tilde{Y}\rightarrow Y$, since $Y=
    \tilde{Y}-m\indicator_{\{\tilde{Y}>m\}}$ and, consequently,
    $I(F ;\tilde{Y})\geq I(F ;Y)$. Furthermore, the reader can verify that
    \begin{equation*}
      \min_{Y\rightarrow \hat{ F }} \Pr\{F \neq \hat{F }\}=
      \min_{\tilde{Y}\rightarrow \hat{F }} \Pr\{F \neq \hat{F }\}=\alpha.
    \end{equation*}
    In particular, given the optimal estimator $\tilde{Y} \rightarrow \hat{F }$, a detection error can
    only occur when  $\tilde{Y}\in \{k+1,\dots, m\}$, in which case $\hat{F }=F $
    with probability 1/2.
    
    Finally,
    \begin{align*}
      H( F|\tilde{Y}) &=- \sum_{\substack{b\in \{-1,1\}\\y\in [2m]}} p_{\tilde{Y}
    } (y)
        p_{F|\tilde{Y}}(b|y) \log p_{F|\tilde{Y}}(b|y) \\
        &= \sum_{y\in \{m+1,2m\}} p_{\tilde{Y} } (y)\\
        &\geq 2 \alpha.
    \end{align*}
    Consequently, $I(F;\tilde{Y})=H(F)-H(F|\tilde{Y})\leq 1-2\alpha$. The result
    follows.
\end{proof}

Let $X^n$ be a plaintext message composed by a sequence of $n$ bits drawn from
$\{-1,1\}^n$. The plaintext can be perfectly hidden by using a one-time pad: A
ciphertext $Y^n$ is produced as $Y^n=X^n\otimes Z^n$, where the key $K=Z^n$ is a
uniformly distributed sequence of $n$  i.i.d. bits chosen independently from
$X^n$. The one-time pad is impractical since, as mentioned, it requires Alice and Bob
to share a very long key.  

Instead of trying to hide the entire plaintext message, assume that Alice and
Bob wish to hide only a set of functions of the plaintext  from Eve. In
particular, we denote this set of functions as $\Phi = \{\phi_1,\dots,\phi_m\}$ where
$\phi_i:\{-1,1\}^n\rightarrow \{-1,1\}$, $\EE{\phi_i(X^n)}=0$ and
$\EE{\phi_i(X^n)\phi_j(X^n)}=0$. The set of functions $\Phi$ is said to be  hidden
$I(\phi_i(X^n);Y^n)=0$ for all $\phi_i\in \Phi$. Can this be accomplished with
a key that satisfies $H(K)\ll H(X^n)$?

The answer is positive, but it depends on $\Phi$. 
We denote the Fourier expansion of $\phi_i\in \Phi$  as
\begin{equation*}
  \phi_i = \sum_{\calS\subseteq [n]} \rho_{i,\calS}\chi_\calS.
\end{equation*}
The following result shows that $\phi_i$ is perfectly hidden from Eve if and only
if $I(\chi_\calS(X^n);Y^n)=0$ for all $\chi_\calS$ such that $\rho_{i,\calS}>
0$.

\begin{lem}
  If $I(\phi_i(X^n);Y^n)=0$ for all $\phi_i\in \Phi$, then $I(\chi_{\calS}(X^n);Y^n)=0$
  for all $\calS$ such that $\rho_{i,\calS}>0$ for some $i\in[m]$.
\end{lem}
\begin{proof}
   Assume that $I(\chi_{\calS}(X^n);Y^n)>0$ for a given $\rho_{i,\calS}>0$. Then
   there exists $b:\calY^n\to \{-1,1\}$ such that
   $\EE{b(Y^n)\chi_\calS(X^n)}=\lambda> 0$.
   Consequently, from \eqref{eq:sharp}, $\EE{b(Y^n)\phi_1(X^n)}\geq \lambda
   \rho_{i,\calS}>0$, and $\phi_1(X^n)$ is not independent of $Y^n$.
\end{proof}

The previous result shows that hiding a set of functions perfectly, or even a
single function, might be as hard as hiding $X^n$. Indeed, if there is a
$\phi_i\in \Phi$ such that $\EE{\phi_i(X^n)\calX_\calS(X^n)}>0$ for all $\calS\subseteq [n]$
where $|\calS|=1$, then perfectly hiding this set of functions can only be
accomplished by using a one-time pad. Nevertheless, if we step back from perfect
secrecy, a large class of functions can be hidden with a comparably small key, as
in the next example.

\begin{example}[BSC revisited]
  Let $Z^n$ be a sequence of $n$ i.i.d. bits such that
  $\Pr\{Z_i=-1\}=\epsilon$, and consider once again the one-time pad
  $Y^n=X^n\otimes Z^n$. Furthermore, denote
  \begin{align*}
    \Phi_k = \left\{ \phi:\{-1,1\}^n \rightarrow\{-1,1\}\suchthat \EE{\phi(X^n)\chi_\calS(X^n)}=0~\forall |\calS|<k  \right\}.
  \end{align*}
  Let $\phi\in \Phi_k$ and $\phi(X^n)=\sum_{\calS:|\calS|\geq k} \rho_\calS
  \chi_\calS(X^n)$. Then, from  Theorem \ref{prop:tighter} and Corollary
  \ref{prop:onebit}, for any $\hat{b}:\calY^n\to \{-1,1\}$,
  \begin{align*}
    \Pr\{\phi(X^n)\neq \hat{b}(Y^n)\}&\geq
    \frac{1}{2}\left(1-\sqrt{\sum_{|\calS|>T}
    \rho_\calS^2(1-2\epsilon)^{2|\calS|}}\right)\\
    &\geq \frac{1}{2}\left( 1-(1-2\epsilon)^k \right).
  \end{align*}
    Consequently, from Lemma \ref{lem:erasure}, $I(\phi(X^n);Y^n)\leq
    (1-2\epsilon)^k$ for all $\phi\in \Phi_k$. Note that $H(Z^n)=nh(\epsilon)$,
    which can be made very small compared to $n$. Therefore, even with a small
    key, a large class of functions can be almost perfectly hidden from the
    eavesdropper through this simple one-time pad scheme. The BSC setting
    discussed in Example \ref{example:noise} is
    generalized in the following theorem which, in turn, is a particular case of
    the analysis in \cite{calmon_exploration_2014}.
\end{example}

\begin{thm}[Generalized One-time Pad]
  Let $Y^n=X^n\otimes Z^n$, $X^n\independent Z^n$, $X^n$ be uniformly distributed, $\phi :\{-1,1\}^n\rightarrow
  \{-1,1\}$ and $\phi (X^n)=\sum_{\calS\subseteq[n]}\rho_\calS
  \chi_\calS(X^n)$. We define $c_\calS\defined \EE{\chi_\calS(Z^n)}$ for $\calS\subseteq [n]$. Then
  \begin{equation}
    \label{eq:otp_general}
    I(\phi (X^n);Y^n)\leq \sqrt{\sum_{\calS\subseteq [n]} (c_\calS\rho_\calS)^2}.
  \end{equation}
  In particular, $I(\phi (X^n);Y^n)=0$ if and only if $c_\calS=0$ for all $\calS$
  such that $\rho_\calS\neq 0$.  
\end{thm}
\begin{proof}
  Let $\psi :\{-1,1\}^n\rightarrow \{-1,1\}$ and 
  $\psi (Y^n)=\sum_{\calS\subseteq[n]} d_\calS\chi_\calS(Y^n)$. Note that
  $\sum_{\calS\subseteq [n]} d_\calS^2=1$. Then
  \begin{align}
    \EE{\phi (X^n)\psi (Y^n)}&=\EE{\phi (X^n)\EE{\psi (Y^n)|X^n}} \nonumber\\
                     &= \EE{\phi (X^n) \sum_{\calS\subseteq[n]}d_\calS
                     \EE{\chi_\calS(Y^n)|X^n}} \nonumber\\
                     &=\EE{\phi (X^n) \sum_{\calS\subseteq[n]}d_\calS
                     \EE{\chi_\calS(X^n\otimes Z^n)|X^n}} \nonumber\\ 
                     &=\EE{\phi (X^n) \sum_{\calS\subseteq[n]}d_\calS
                     \EE{\chi_\calS(X^n)\chi_\calS(Z^n)|X^n}} \nonumber\\
                     &= \sum_{\calS\subseteq[n]}d_\calS
                     \EE{\phi (X^n)\chi_\calS(X^n)}\EE{
                     \chi_\calS(Z^n)} \nonumber \\
                     &=  \sum_{\calS\subseteq[n]}d_\calS
                     \rho_\calS c_\calS \label{eq:proof_otp}\\
                     &\leq \sqrt{\sum_{\calS\subseteq [n]}
                   (c_\calS\rho_\calS)^2},  \label{eq:CSproof}                   
  \end{align}
  where \eqref{eq:CSproof} follows from the Cauchy-Schwarz inequality. The
  inequality \eqref{eq:otp_general} then follows from Lemma \ref{lem:erasure}.
  Finally, assume there exists $\calS\subseteq [n]$ such that both $c_\calS\neq
  0$ and $\rho_\calS\neq 0$. Then setting $\psi (Y^n)=\chi_\calS(Y^n)$, it follows
  from \eqref{eq:proof_otp} that $\EE{\phi (X^n)\psi (Y^n)}=\rho_\calS c_\calS \neq 0$
  and, consequently, $I(\phi (X^n);Y^n)> 0$.
\end{proof}

\subsection{From Symbol Secrecy to Function Secrecy}
\label{sec:from_symbol}

Symbol secrecy captures the amount of information that an encryption scheme
leaks about individual symbols of a message. A given encryption scheme can
achieve a high level of (weak) information-theoretic security, but low symbol
secrecy. As illustrated in  Section \ref{sec:trivial}, by sending a constant
fraction of the message in the clear, the average amount of information about
the plaintext that leaks  relative to the length of the message can be made
arbitrarily small, nevertheless the symbol secrecy performance is always
constant (i.e. does not decrease with message length).


When $X$ is uniformly drawn from $\mathbb{F}_q$ for which an $(n,k,n-k+1)$ MDS code
exists, then an absolute symbol secrecy of $k/n$ can always be achieved
using the encryption scheme suggested in Proposition \ref{prop:MDS}. If $X$ is a
binary random variable, then we can map sequences of plaintext bits of length
$\floor{\log_2 q}$ to an appropriate symbol in $\mathbb{F}_q$, and then use the
parity check matrix of an  MDS code to achieve a high symbol secrecy. Therefore,
we may assume without  loss of generality  that $X^n$ is drawn from
$\{-1,1\}^n$. We also make the assumption that $X^n$ is uniformly distributed.
This can be regarded as an approximation for the distribution of $X^n$ when it
is, for example, the output of an optimal source encoder with sufficiently large
blocklength.

\begin{thm}
  Let $X^n$ be a uniformly distributed sequence of $n$ bits, $Y=\Enc_n(X^n,K)$,
  and $u_\epsilon$ and $\epsilon_t^*$ the corresponding
  symbol secrecy and dual symbol secrecy of $\Enc_n$, defined in
  \eqref{eq:def_ssecrecy} and \eqref{eq:def_estar}, respectively. Furthermore, for $\phi :
  \{-1,1\}^n\rightarrow \{-1,1\}$ and $\EE{\phi (X^n)}=0$, let $\phi (X^n)=\sum_{\calS\subseteq
  [n]}\rho_\calS\chi_\calS(X^n)$.  Then for any $\hat{\phi }:\calY\rightarrow
  \{-1,1\}$
  \begin{equation}
    \Pr\{\phi (X^n)\neq \hat{\phi }(Y)\}\geq \frac{1}{2}\left( 1-B_{|\Phi|}(\brho,\blambda)
    \right),
  \end{equation}
  where $\Phi=\{\chi_\calS:\rho_\calS\neq 0 \}$, $\lambda(t)\defined
  h_b^{-1}((1-\epsilon^*_{t}t)^+)$, $\blambda =
  \{\lambda(|\calS|)\}_{\calS\subseteq[n]}$ and $\brho = \{|\rho_\calS|\}_{\calS\subseteq
  [n]}$. In particular,
  \begin{equation}
    \Pr\{\phi (X^n)\neq \hat{\phi }(Y)\}\geq \frac{1}{2}\left(
    1-\sqrt{\sum_{|\calS|>n\mu_0}\rho_\calS^2} \right).
  \end{equation}
\end{thm}
\begin{proof}
    From the definition of symbol secrecy, for any
    $\calS \subseteq  [n]$ with $|\calS|=t$
    \begin{align*}
      I(\chi_\calS(X^n);Y)\leq I(X^{\calS};Y)\leq
      \epsilon^*_{t}t,
    \end{align*}
    and, consequently,
    \begin{align*}
      H(\chi_\calS(X^n)|Y)\geq (1-\epsilon^*_{t}t)^+.
    \end{align*}
    From Fano's inequality, for any binary $\hat{F}$ where $Y\rightarrow
    \hat{F}$
    \begin{align*}
      \Pr\{\chi_\calS(X^n)\neq \hat{F}\}\geq
      h_b^{-1}((1-\epsilon^*_{t}t)^+ ),
    \end{align*}
    where $h_b^{-1}:[0,1]\rightarrow[0,1/2]$ is the inverse of the binary
    entropy function. In particular, from the definition of absolute symbol
    secrecy, if $\epsilon^*_t = 0$, then
    \begin{align*}
        \Pr\{\chi_\calS(X^n)\neq \hat{F}\}=1/2~\forall |\calS|\leq n\mu_0.
    \end{align*}
    The result then follows directly from Theorem \ref{prop:tighter}, the
    fact that $\phi (X^n)=\sum_{\calS\subseteq [n]} \rho_\calS \chi_\calS(X^n)$
    and letting $\lambda(t)\defined
    h_b^{-1}((1-\epsilon^*_{t}t)^+ )$.
  \end{proof}

\section{Discussion}

\label{sec:practical}

In this section we discuss the application of our results to different settings
in privacy and cryptography.

\subsection{The Correlation-Error Product} 
\label{sec:privacy}

We momentarily diverge from the cryptographic setting and
introduce the \textit{error-correlation product} for the privacy setting
considered by Calmon and Fawaz in \cite{inproc:allertonPriv}. Let $W$ and $X$ be
two random variables with joint distribution $p_{W,X}$. $W$ represents a
variable that is supposed to remain private, while $X$ represents a variable
that will be released to an untrusted data collector in order to receive some
utility based on $X$. The goal is to design a randomized mapping $p_{Y|X}$,
called the privacy assuring mapping, that transforms $X$ into an output $Y$ that
will be disclosed to a third party.

The goal of a privacy assuring mechanism  is to produce an output $Y$, derived
from $X$ according to the mapping $p_{Y|X}$, that will be released to the data
collector in the place of $X$. The released variable $Y$ is chosen such that $W$
cannot be inferred reliably given an observation of $Y$.  Simultaneously, given
an appropriate distortion metric, $X$ should be close enough to $Y$ so that a
certain level of utility can still be provided.  For example, $W$ could be a
user's political preference, and $X$ a set of movie ratings released to a
recommender system in order to receive movie recommendations. $Y$ is chosen as a
perturbed version of the movie recommendations so that the user's political
preference is obscured, while meaningful recommendations can still be provided.

Given $W\rightarrow X\rightarrow Y$ and $p_{W,X}$, a  privacy
assuring mapping is given by the conditional distribution $p_{Y|X}$. The
choice of $p_{Y|X}$ determines the tradeoff between privacy and utility. If
$p_{Y|X}=p_{Y}$, then perfect privacy is achieved (i.e. $W$ and $Y$ are
independent), but no utility can be
provided. Conversely, if $p_{Y|X}$ is the identity mapping, then no privacy is
gained, but the highest level of utility can be provided.  

When $W=\phi(X)$ where $\phi\in \calL_2(p_X)$, the bounds from Section \ref{sec:MMSE} shed light on the
fundamental  privacy-utility tradeoff. Returning to the notation of
Section \ref{sec:MMSE}, let $W=\phi(X)$ be correlated
with a set of functions $\Phi=\{\phi_i\}_{i=1}^m$. The next result is a direct
corollary of Theorem \ref{prop:tighter}.
\begin{cor}
    Let $\EE{W\phi_i(X)}=\rho_i$, $\sum_{i=1}^{|\Phi|}\rho_i^2=1$,
    $\psi_i(Y)=\EE{\phi_i(X)|Y}$ and, for $i\neq j$, $\EE{\phi_i(X)\phi_j(X)}=0$
    and $\EE{\psi_i(Y)\psi_j(Y)}=0$. Then
    \begin{equation}
        \mmse(W|Y)= \sum_{i=1}^{|\Phi|}\mmse(\phi_i(Y)|X)\rho_i^2.
    \end{equation}
\end{cor}
We call the product $\mmse(\phi_i(Y)|X)\rho_i^2$ the \textit{error-correlation}
product. The secret variable $W$  cannot be estimated with low MMSE from $Y$ if
and only if the functions $\phi _i$ that are strongly correlated with $W$ (i.e.
large $\rho_i^2$) cannot be estimated reliably. Consequently, if $\rho_i$ is
large and $\phi _i$ is relevant for the utility provided by the data collector,
privacy cannot be achieved without a significant loss of utility: $\mmse(\phi
_i(X)|Y)$ is necessarily large if $\mmse(W|Y)$ is large. Conversely, in order to
hide $W$, it is sufficient to hide the functions $\phi _i(X)$ that are strongly
correlated with $\phi (X)$. This no-free-lunch result is intuitive, since one
would expect that privacy cannot be achieved if utility is based on data that is
strongly correlated with the private variables. The results presented here prove
that this is indeed the case.

We present next a general description of a two-phase secure
communication scheme for the threat model described in Section \ref{sec:model},
presented in terms of the list-source code constructions derived using linear
codes. Note that this scheme can be easily extended to any list-source code by
using the corresponding encoding/decoding functions instead of multiplication by
parity check matrices.

\subsection{A Secure Communication Scheme Based on List-Source Codes}

We assume that Alice and Bob have access to a symmetric-key encryption/decryption scheme
$(\mathsf{Enc}',\mathsf{Dec}')$ that is used with the shared secret key $K$ and
is sufficiently secure against the adversary. This scheme can be, for example, a
one-time pad. The encryption/decryption procedure is performed as follows, and
will be used as components of  the overall encryption scheme
$(\mathsf{Enc},\mathsf{Dec})$ described below. 

\begin{scheme}
\label{scheme:prac}

\textit{Input}: The source encoded sequence $\bx\in \Fq^n$, parity check matrix $\bH$ of
a linear code in $\Fq^n$, a full-rank $k\times n$  matrix $\bD$ such that
$\rank([\bH^T~\bD^T])=n$, and encryption/decryption functions
$(\mathsf{Enc'},\mathsf{Dec'})$. We assume both Alice and Bob share a secret key
$K$.

\noindent \textbf{Encryption} $(\Enc)$:

\noindent \textit{Phase I (pre-caching)}: Alice generates $\bsigma = \bH
\bx$ and  sends to Bob.\footnote{Here, Alice can use message authentication codes and
public key encryption to augment security. Furthermore, the list-source coding
scheme can be used as an additional layer of security  with information-theoretic
guarantees in symmetric-key ciphers. Since we are interested  in the
information-theoretic security properties of the scheme, we will not go into
further details. We do recognize that in order to use this scheme in practice additional
steps are needed to meet modern cryptographic standards.}

\noindent \textit{Phase II (send encrypted data)}: Alice generates
$\mathbf{e}=\mathsf{Enc'}(\bD \bx,K)$ and sends to Bob.

\noindent \textbf{Decryption} $(\Dec)$: Bob calculates $
\Dec'(\mathbf{e},K)=\bD \bx$ and
recovers $\bx$ from $\bsigma$ and $\bD \bx $.

\end{scheme}

Assuming that  $(\mathsf{Enc'},\mathsf{Dec'})$ is secure, the
information-theoretic security of Scheme
\ref{scheme:prac} reduces to the security of the underlying list-source code
(i.e. Scheme \ref{scheme:lin}). In practice, the encryption/decryption functions
$(\mathsf{Enc'},\mathsf{Dec'})$ may depend on a secret or public/private key, as
long as it provide sufficient security for the desired application.  In
addition, assuming that the source sequence is uniform and i.i.d. in $ F_q^n$,
we can use MDS codes to make strong security guarantees, as described in the
next section. In this case, an adversary that observes $\bsigma$ cannot infer
\textit{any} information about any set of $k$ symbols of the original message. 

 Note that this scheme
 has a  \textit{tunable} level of secrecy: The amount of data sent in phase I and
phase II can be appropriately selected to match the properties of the encryption
scheme available, the size of the key length, and the desired level of secrecy.
Furthermore, when the encryption procedure has a higher computational cost than
the list-source encoding/decoding operations, list-source codes can be used to
reduce the total number of operations required by allowing encryption of a
smaller portion of the message (phase II).

The protocol outline presented in Scheme \ref{scheme:prac}  is useful in
different practical scenarios, which are discussed in the following sections.
Most of the advantages of the suggested scheme stem from the fact that
list-source codes are key-independent, allowing content to be distributed  when
a key distribution infrastructure is not yet established, and providing an
additional level of security if keys are compromised before phase II in Scheme
\ref{scheme:prac}.

\subsection{Content pre-caching}

As hinted earlier, list-source codes provide a secure mechanism for content
pre-caching when a key infrastructure has not yet been established. A large
fraction of the data can be  list-source coded and  securely transmitted before
the termination of the key distribution protocol. This is particularly
significant in large networks with hundreds of mobile nodes, where key
management protocols can require a significant amount of time to complete
\cite{eschenauer_key-management_2002}. Scheme \ref{scheme:prac} circumvents the
communication delays incurred by key compromise detection, revocation and
redistribution by allowing data to be efficiently distributed concurrently with
the key distribution protocol, while maintaining a level of security determined
by the underlying list-source code.

\subsection{Application to key distribution protocols}

List-source codes can also provide additional robustness to key compromise. If
the secret key is compromised before phase II of Scheme  \ref{scheme:prac}, the
data will still be as secure as the underlying list-source code. Even if a
(computationally unbounded) adversary has perfect knowledge of the key, until
the last part of the data is transmitted the best he can do is reduce the number
of possible inputs to an exponentially large list. In contrast, if a stream
cipher based on a pseudo-random number generator were used and the
initial seed was leaked to an adversary, all the data transmitted up to the
point where the compromise was detected would be vulnerable. The use of
list-source codes provide an additional, information-theoretic level of security
to the data up to the point where the last fraction of the message is
transmitted.  This also allows decisions as to which receivers will be allowed
to decrypt the data can be delayed until the very end of the transmission,
providing more time for detection of unauthorized receivers and allowing a
larger flexibility in key distribution.

In addition, if the level of security provided by the list-source code is
considered sufficient and the key is compromised before phase II, the key can be
redistributed \textit{without the need of retransmitting the entire data}. As
soon as the keys are reestablished, the transmitter simply encrypts the
remaining part of the data  in phase II with the new key.

\subsection{Additional layer of security}

We also highlight that list-source codes can be used to provide an additional
layer of security to the underlying encryption scheme. The message can be
list-source coded  after encryption and transmitted in two phases, as in Scheme 
\ref{scheme:prac}. As argued in the previous point, this provides additional
robustness against key compromise, in particular when a compromised key can
reveal a large amount of information about an incomplete message (e.g. stream
ciphers). Consequently, list-source codes are a simple, practical way of
augmenting the security of current  encryption schemes.

One example application is to combine list-source codes with stream ciphers. The
source-coded message can be initially encrypted using a pseudorandom number generator
(PRG) initialized with a randomly selected seed, and then list-source coded. The
initial random seed would be part of the encrypted message sent in the final
transmission phase. This setup has the advantage of augmenting the security of
the underlying stream cipher, and provides randomization to the list-source
coded message. In particular, if the LSC is based on MDS codes and assuming that
the distribution of the plaintext is nearly uniform, strong
information-theoretic symbol secrecy guarantees can be made about the
transmitted data, as discussed in Section \ref{sec:symbolsecrecy}.  Even if the
underlying PRG is compromised, the message would still be secure.

\subsection{Tunable level of secrecy }

List-source codes provide a tunable level of secrecy, i.e. the amount of
security provided by the scheme can be adjusted according to the application of
interest. This can be done by appropriately selecting the size of the list ($L$)
of the underlying code, which determines the amount of uncertainty an adversary
will have regarding the input message. In the proposed implementation using
linear codes, this corresponds to choosing the size of the parity check matrix
$\mathbf{H}$, or, analogously, the parameters of the underlying error-correcting
code. In terms of Scheme  \ref{scheme:prac}, a larger (respectively smaller)
value of $L$ will lead to a smaller (larger) list-source coded message in phase
I and a larger (smaller) encryption burden in phase II.

\section{Conclusions}
\label{sec:conc}

We conclude the paper with a summary of our contributions. We introduce the
concept of LSCs, which are codes that compress a source below its
entropy rate. We derived fundamental bounds for the rate list region, and
provided code constructions that achieve these bounds.  List-source codes are a
useful tool for understanding how to perform encryption when the (random) key
length is smaller than the message entropy. When the key is
small, we can reduce an adversary's uncertainty to a near-uniformly distributed
list of possible source sequences with an exponential (in terms of the key
length) number of elements by using list-source codes. We also demonstrated how
list-source codes can be implemented using standard linear codes. 

Furthermore, we presented a new information-theoretic metric of secrecy, namely
$\epsilon$-symbol secrecy, which characterizes the amount of information leaked
about specific symbols of the source given an encoded version of the message.
We derived fundamental bounds for  $\epsilon$-symbol secrecy, and showed how
these bounds can be achieved using MDS codes when the source is uniformly
distributed.

We also introduced results for bounding the probability that an adversary
correctly guesses a predicate of the plaintext in terms of the symbol secrecy
achieved by the underlying encryption scheme. These results are based on Lemma
\ref{lem:quadBound}, which, in turn, was used to derive bounds on the information leakage of a
security system that does not achieve perfect secrecy. These bounds provide
insight on how to design symmetric-key encryption schemes that hide specific
functions of the data, where uncertainty is captured in terms of minimum-mean
squared error. These results also shed light on the fundamental privacy-utility
tradeoff in privacy systems.

\bibliographystyle{IEEEtran}
\bibliography{IEEEabrv,references}

\begin{thebibliography}{10}
\providecommand{\url}[1]{#1}
\csname url@samestyle\endcsname
\providecommand{\newblock}{\relax}
\providecommand{\bibinfo}[2]{#2}
\providecommand{\BIBentrySTDinterwordspacing}{\spaceskip=0pt\relax}
\providecommand{\BIBentryALTinterwordstretchfactor}{4}
\providecommand{\BIBentryALTinterwordspacing}{\spaceskip=\fontdimen2\font plus
\BIBentryALTinterwordstretchfactor\fontdimen3\font minus
  \fontdimen4\font\relax}
\providecommand{\BIBforeignlanguage}[2]{{%
\expandafter\ifx\csname l@#1\endcsname\relax
\typeout{** WARNING: IEEEtran.bst: No hyphenation pattern has been}%
\typeout{** loaded for the language `#1'. Using the pattern for}%
\typeout{** the default language instead.}%
\else
\language=\csname l@#1\endcsname
\fi
#2}}
\providecommand{\BIBdecl}{\relax}
\BIBdecl

\bibitem{inproc:allertonCrypt}
F.~P. Calmon, M.~M\'edard, L.~Zeger, J.~Barros, M.~M. Christiansen, and K.~R.
  Duffy, ``Lists that are smaller than their parts: A coding approach to
  tunable secrecy,'' in \emph{Proc. 50th Annual Allerton Conf. on Commun.,
  Control, and Comput.}, 2012.

\bibitem{calmon2014allerton}
F.~P. Calmon, M.~Varia, and M.~M\'edard, ``On information-theoretic metrics for
  symmetric-key encryption and privacy,'' in \emph{Proc. 52nd Annual Allerton
  Conf. on Commun., Control, and Comput.}, 2014.

\bibitem{shannon_communication_1949}
C.~E. Shannon, ``Communication theory of secrecy systems,'' \emph{Bell {S}ystem
  {T}echnical {J}ournal}, vol.~28, no.~4, pp. 656--715, 1949.

\bibitem{liang_information_2009}
Y.~Liang, H.~V. Poor, and S.~Shamai~{(Shitz)}, ``Information theoretic
  security,'' \emph{Found. Trends Commun. Inf. Theory}, vol.~5, pp. 355--580,
  Apr. 2009.

\bibitem{katz_introduction_2007}
J.~Katz and Y.~Lindell, \emph{Introduction to Modern Cryptography: Principles
  and Protocols}, 1st~ed.\hskip 1em plus 0.5em minus 0.4em\relax Chapman and
  {Hall/CRC}, Aug. 2007.

\bibitem{hellman_extension_1977}
M.~Hellman, ``An extension of the {S}hannon theory approach to cryptography,''
  \emph{{IEEE} Trans. Inf. Theory}, vol.~23, no.~3, pp. 289--294, May 1977.

\bibitem{diffie_new_1976}
W.~Diffie and M.~Hellman, ``New directions in cryptography,'' \emph{{IEEE}
  Trans. Inf. Theory}, vol.~22, no.~6, pp. 644--654, Nov. 1976.

\bibitem{blahut_communications_1994}
R.~E. Blahut, D.~J. Costello, U.~Maurer, and T.~Mittelholzer, Eds.,
  \emph{Communications and Cryptography: Two Sides of One Tapestry},
  1st~ed.\hskip 1em plus 0.5em minus 0.4em\relax Springer, Jun. 1994.

\bibitem{goldwasser_probabilistic_1984}
S.~Goldwasser and S.~Micali, ``Probabilistic encryption,'' \emph{Journal of
  Computer and System Sciences}, vol.~28, no.~2, pp. 270--299, Apr. 1984.

\bibitem{inproc:allertonPriv}
F.~P. Calmon and N.~Fawaz, ``A framework for privacy against statistical
  inference,'' in \emph{Proc. 50th Ann. Allerton Conf. Commun., Contr., and
  Comput.}, Oct. 2012.

\bibitem{ahlswede_remarks_1982}
R.~Ahlswede, ``Remarks on shannon's secrecy systems,'' \emph{Problems of
  Control and Inf. Theory}, vol.~11, no.~4, 1982.

\bibitem{lu_existence_1979}
S.-C. Lu, ``The existence of good cryptosystems for key rates greater than the
  message redundancy (corresp.),'' \emph{{IEEE} Trans. Inf. Theory}, vol.~25,
  no.~4, pp. 475--477, Jul. 1979.

\bibitem{lu_random_1979}
------, ``Random ciphering bounds on a class of secrecy systems and discrete
  message sources,'' \emph{{IEEE} Trans. Inf. Theory}, vol.~25, no.~4, pp.
  405--414, Jul. 1979.

\bibitem{lu_secrecy_1979}
------, ``On secrecy systems with side information about the message available
  to a cryptanalyst (corresp.),'' \emph{{IEEE} Trans. Inf. Theory}, vol.~25,
  no.~4, pp. 472--475, Jul. 1979.

\bibitem{schieler_rate-distortion_2014}
C.~Schieler and P.~Cuff, ``Rate-distortion theory for secrecy systems,''
  \emph{{IEEE} Trans. Inf. Theory}, vol.~{PP}, no.~99, 2014.

\bibitem{ozarow_wire-tap_1985}
L.~Ozarow and A.~Wyner, ``Wire-tap channel {II},'' in \emph{Advances in
  Cryptology}, 1985, pp. 33--50.

\bibitem{cai_secure_2002}
N.~Cai and R.~Yeung, ``Secure network coding,'' in \emph{Proc. {IEEE} Int.
  Symp. on Inf. Theory}, 2002.

\bibitem{feldman_capacity_2004}
J.~Feldman, T.~Malkin, C.~Stein, and R.~A. Servedio, ``On the capacity of
  secure network coding,'' in \emph{Proc. 42nd Ann. Allerton Conf. Commun.,
  Contr., and Comput.}, 2004.

\bibitem{mills_secure_2008}
A.~Mills, B.~Smith, T.~Clancy, E.~Soljanin, and S.~Vishwanath, ``On secure
  communication over wireless erasure networks,'' in \emph{Proc. {IEEE} Int.
  Symp. on Inf. Theory}, Jul. 2008, pp. 161 --165.

\bibitem{el_rouayheb_secure_2012}
S.~El~Rouayheb, E.~Soljanin, and A.~Sprintson, ``Secure network coding for
  wiretap networks of type {II},'' \emph{{IEEE} Trans. Inf. Theory}, vol.~58,
  no.~3, pp. 1361 --1371, Mar. 2012.

\bibitem{silva_universal_2011}
D.~Silva and F.~Kschischang, ``Universal secure network coding via
  {Rank-Metric} codes,'' \emph{{IEEE} Trans. Inf. Theory}, vol.~57, no.~2, pp.
  1124 --1135, Feb. 2011.

\bibitem{lima_random_2007}
L.~Lima, M.~M\'edard, and J.~Barros, ``Random linear network coding: A free
  cipher?'' in \emph{Proc. {IEEE} Int. Symp. on Inf. Theory}, Jun. 2007, pp.
  546 --550.

\bibitem{cai_theory_2011}
N.~Cai and T.~Chan, ``Theory of secure network coding,'' \emph{{IEEE} Proc.},
  vol.~99, no.~3, pp. 421 --437, Mar. 2011.

\bibitem{oliveira_trusted_2010}
P.~Oliveira, L.~Lima, T.~Vinhoza, J.~Barros, and M.~M\'edard, ``Trusted storage
  over untrusted networks,'' in \emph{{IEEE} Global Telecommunications
  Conference}, Dec. 2010, pp. 1 --5.

\bibitem{elias_list_1957}
P.~Elias, ``List decoding for noisy channels,'' Research Laboratory of
  Electronics, MIT, Technical Report 335, September 1957.

\bibitem{wozencraft1958}
J.~M. Wozencraft, ``List decoding,'' Research Laboratory of Electronics, MIT,
  Progress Report~48, 1958.

\bibitem{shannon_lower_1967}
C.~Shannon, R.~Gallager, and E.~Berlekamp, ``Lower bounds to error probability
  for coding on discrete memoryless channels. {I},'' \emph{Information and
  Control}, vol.~10, no.~1, pp. 65--103, Jan. 1967.

\bibitem{shannon_lower_1967-1}
------, ``Lower bounds to error probability for coding on discrete memoryless
  channels. {II},'' \emph{Information and Control}, vol.~10, no.~5, pp.
  522--552, May 1967.

\bibitem{forney_exponential_1968}
G.~Forney, ``Exponential error bounds for erasure, list, and decision feedback
  schemes,'' \emph{{IEEE} Trans. Inf. Theory}, vol.~14, no.~2, pp. 206-- 220,
  Mar. 1968.

\bibitem{guruswami_list_2001}
V.~Guruswami, ``List decoding of error-correcting codes,'' Thesis, MIT,
  Cambridge, MA, 2001.

\bibitem{guruswami_list_2009}
------, ``List decoding of binary {Codes--A} brief survey of some recent
  results,'' in \emph{Coding and Cryptology}, ser. Lecture Notes in Computer
  Science.\hskip 1em plus 0.5em minus 0.4em\relax Springer Berlin / Heidelberg,
  2009, vol. 5557, pp. 97--106.

\bibitem{ali_source_2010}
M.~Ali and M.~Kuijper, ``Source coding with side information using list
  decoding,'' in \emph{Proc. {IEEE} Int. Symp. on Inf. Theory}.\hskip 1em plus
  0.5em minus 0.4em\relax {IEEE}, Jun. 2010, pp. 91--95.

\bibitem{yamamoto_rate-distortion_1997}
H.~Yamamoto, ``Rate-distortion theory for the {Shannon} cipher system,''
  \emph{{IEEE} Trans. Inf. Theory}, vol.~43, no.~3, pp. 827--835, May 1997.

\bibitem{reed_information_1973}
I.~S. Reed, ``Information {Theory} and {Privacy} in {Data} {Banks},'' in
  \emph{Proceedings of the {June} 4-8, 1973, {National} {Computer} {Conference}
  and {Exposition}}, ser. {AFIPS} '73.\hskip 1em plus 0.5em minus 0.4em\relax
  New York, NY, USA: ACM, 1973, pp. 581--587.

\bibitem{sarwate_rate-disortion_2014}
A.~Sarwate and L.~Sankar, ``A rate-disortion perspective on local differential
  privacy,'' in \emph{Proc. 52nd Ann. Allerton Conf. Commun., Contr., and
  Comput.}, Sep. 2014, pp. 903--908.

\bibitem{rebollo-monedero_t-closeness-like_2010}
D.~Rebollo-Monedero, J.~Forne, and J.~Domingo-Ferrer, ``From
  t-{Closeness}-{Like} {Privacy} to {Postrandomization} via {Information}
  {Theory},'' \emph{IEEE Trans. on Knowledge and Data Engineering}, vol.~22,
  no.~11, pp. 1623--1636, Nov. 2010.

\bibitem{sankar_utility-privacy_2013}
L.~Sankar, S.~Rajagopalan, and H.~Poor, ``Utility-{Privacy} {Tradeoffs} in
  {Databases}: {An} {Information}-{Theoretic} {Approach},'' \emph{IEEE Trans.
  on Information Forensics and Security}, vol.~8, no.~6, pp. 838--852, Jun.
  2013.

\bibitem{wyner_wire-tap_1975}
A.~D. Wyner, ``\BIBforeignlanguage{en}{The {Wire}-{Tap} {Channel}},''
  \emph{\BIBforeignlanguage{en}{Bell System Technical Journal}}, vol.~54,
  no.~8, pp. 1355--1387, Oct. 1975.

\bibitem{maurer_information-theoretic_2000}
U.~Maurer and S.~Wolf, ``\BIBforeignlanguage{en}{Information-{Theoretic} {Key}
  {Agreement}: {From} {Weak} to {Strong} {Secrecy} for {Free}},'' in
  \emph{\BIBforeignlanguage{en}{Advances in {Cryptology} {(EUROCRYPT)}}}, ser.
  Lecture {Notes} in {Computer} {Science}, B.~Preneel, Ed.\hskip 1em plus 0.5em
  minus 0.4em\relax Springer Berlin Heidelberg, 2000, no. 1807, pp. 351--368.

\bibitem{csiszar_information_2011}
I.~Csisz\'ar and J.~K\"orner, \emph{Information Theory: Coding Theorems for
  Discrete Memoryless Systems}, 2nd~ed.\hskip 1em plus 0.5em minus 0.4em\relax
  Cambridge University Press, Aug. 2011.

\bibitem{cover_elements_2006}
T.~M. Cover and J.~A. Thomas, \emph{Elements of Information Theory},
  2nd~ed.\hskip 1em plus 0.5em minus 0.4em\relax Wiley-Interscience, Jul. 2006.

\bibitem{roth_introduction_2006}
R.~Roth, \emph{\BIBforeignlanguage{English}{Introduction to {Coding}
  {Theory}}}.\hskip 1em plus 0.5em minus 0.4em\relax Cambridge, UK ; New York:
  Cambridge University Press, Mar. 2006.

\bibitem{ho_network_2008}
T.~Ho and D.~Lun, \emph{\BIBforeignlanguage{English}{Network {Coding}: {An}
  {Introduction}}}.\hskip 1em plus 0.5em minus 0.4em\relax New York: Cambridge
  University Press, Apr. 2008.

\bibitem{odonnell_topics_2008}
R.~{O'Donnell}, ``Some topics in analysis of boolean functions,'' in
  \emph{Proc. 40th {ACM} Symp. on Theory of Computing}, 2008, pp. 569--578.

\bibitem{calmon_exploration_2014}
F.~Calmon, M.~Varia, and M.~M\'edard, ``An exploration of the role of principal
  inertia components in information theory,'' in \emph{Proc. {IEEE} {Inf.}
  {Theory} {Workshop} ({ITW})}, Nov. 2014, pp. 252--256.

\bibitem{eschenauer_key-management_2002}
L.~Eschenauer and V.~D. Gligor, ``A key-management scheme for distributed
  sensor networks,'' in \emph{Proceedings of the 9th {ACM} Conference on
  Computer and Communications Security}, ser. {CCS} '02.\hskip 1em plus 0.5em
  minus 0.4em\relax New York, {NY}, {USA}: {ACM}, 2002, pp. 41--47.

\end{thebibliography}

\appendices
\section{Proof of Lemma \ref{lem:quadBound}}

    For fixed $\ba,\bb \in \Reals^n$ where $a_i>0$ and $b_i\geq0$, let $z_P:\Reals^n\rightarrow \Reals$ and
    $z_D:\Reals^n\rightarrow \Reals$ be given by
    \begin{align*}
        z_P(\by) &\defined \ba^T\by,\\
        z_D(\bu) & \defined \ba^T\bb + \bu^T\bb +\|\bu\|_2.
     \end{align*}  
     Furthermore, we define $\calA(\ba)\defined \left\{ \bu\in
     \Reals^n|\bu\geq\-\ba \right\}$ and $\calB(\bb)\defined \left\{ \by
       \in\Reals^n~|~\|\by\|_2\leq 1,\by\leq\bb \right\}$.

       The optimal value $z_n(\ba,\bb)$ is given by the following pair of primal-dual convex
       programs: 
       \begin{align*}
         z_n(\ba,\bb)=\max_{\by\in\calB(\bb)} z_P(\by)=\min_{\bu\in\calA(\ba)}
         z_D(\bu).
       \end{align*}
       Assume, without loss of generality, that $b_1/a_1\leq
       b_2/a_2\leq\dots\leq b_n/a_n$, and let $k^*$ be defined in
       \eqref{eq:kstar}.  
       
      
       Let $c_j\defined \sqrt{\frac{\left(1-\sum_{i=1}^{j}
     b_i^2\right)}{\|\ba\|_2^2-\sum_{i=1}^{j} a_i^2  }}$. Note that since  $\sum_{i=1}^{k^*}
     b_i^2<1$, we have $c_{k^*}>0$. In addition,  let
     $$\by^*=(b_1,\dots,b_{k^*},a_{k^*+1}c_{k^*},\dots,a_nc_{k^*})$$ and
     $$\bu^*=(-b_1/c_{k^*},\dots,-b_{k^*}/c_{k^*},-a_{k^*+1},\dots,-a_n).$$ From
     the definition of $k^*$, $\by^*\in\calB(\bb)$ and $\bu^*\in
     \calA(\ba)$. Furthermore,
       \begin{align}
            z_P(\by^*)  &= \ba^T\by^*\nonumber\\
            &= \sum_{i=1}^{k^*} a_ib_i +\sum_{i=k^*+1}^n c_{k^*}a_i^2\nonumber\\
            &= \sum_{i=1}^{k^*} a_ib_i + \sqrt{\left(
              \|\ba\|_2^2-\sum_{i=1}^{k^*} a_i^2 \right)\left(1-\sum_{i=1}^{k^*}
          b_i^2\right)}, \label{eq:opt_val}
       \end{align}
       and 
       \begin{align*}
         z_D(\bu^*)  =& \ba^T\bb + {\bu^*}^T\bb + \|\bu^*\|_2\\
         =& \sum_{i=1}^{k^*} \left( a_ib_i -
         \frac{b_i^2}{c_{k^*}}\right)\\
         &+c_{k^*}^{-1}\sqrt{\sum_{i=1}^{k^*}b_i^2+c_{k^*}^2\left(
         \|\ba\|_2^2-\sum_{i=1}^{k^*}a_i^2 \right)}\\
         =& \sum_{i=1}^{k^*} a_ib_i + c_{k^*}^{-1}\left( 1-\sum_{i=1}^{k^*}b_i^2
         \right)\\
         =& \sum_{i=1}^{k^*} a_ib_i + \sqrt{\left( \|\ba\|_2^2-\sum_{i=1}^{k^*} a_i^2 \right)\left(1-\sum_{i=1}^{k^*}
          b_i^2\right)}\\
         =& z_P(\by^*).
       \end{align*}
     Since both the primal and the dual achieve the same value at $\by^*$ and
     $\bu^*$, respectively, it follows that
     the value $z_P(\by^*)$ given in \eqref{eq:opt_val} is optimal.

\end{document}